\documentclass[11pt,oneside,reqno]{article}

\usepackage{amsmath,amsthm,amssymb,color}
\usepackage{bbm}
\usepackage{mathrsfs}
\usepackage[breaklinks=true]{hyperref}
\usepackage{textcase}
\usepackage{graphicx}
\usepackage{caption}
\usepackage{subcaption}

\usepackage[margin=1in]{geometry}

\numberwithin{equation}{section}
\allowdisplaybreaks[4]

\theoremstyle{plain}
\newtheorem{theorem}{Theorem}[section]
\newtheorem{proposition}[theorem]{Proposition}
\newtheorem{lemma}[theorem]{Lemma}

\newtheorem{corollary}[theorem]{Corollary}

\theoremstyle{definition}

\newtheorem{remark}[theorem]{Remark}
\newtheorem{assumption}{Assumption}

\makeatletter

\def\^#1{\ifmmode {\mathaccent"705E #1} \else {\accent94 #1} \fi}
\def\~#1{\ifmmode {\mathaccent"707E #1} \else {\accent"7E #1} \fi}

\def\*#1{#1^\ast}
\edef\-#1{\noexpand\ifmmode {\noexpand\bar{#1}} \noexpand\else \-#1\noexpand\fi}
\def\>#1{\vec{#1}}
\def\.#1{\dot{#1}}

\def\atop{\@@atop}

\renewcommand{\leq}{\leqslant}
\renewcommand{\geq}{\geqslant}
\renewcommand{\phi}{\varphi}
\newcommand{\eps}{\varepsilon}

\renewcommand\section{\@startsection {section}{1}{\z@}%
{-3.5ex \@plus -1ex \@minus -.2ex}%
{1.3ex \@plus.2ex}%
{\center\small\sc\MakeTextUppercase}}

\def\subsection#1{\@startsection {subsection}{2}{0pt}%
{-3.5ex \@plus -1ex \@minus -.2ex}%
{1ex \@plus.2ex}%
{\bf\mathversion{bold}}{#1}}

\def\subsubsection#1{\@startsection{subsubsection}{3}{0pt}%
{\medskipamount}%
{-10pt}%
{\normalsize\itshape}{\kern-2.2ex. #1.}}

\def\blfootnote{\xdef\@thefnmark{}\@footnotetext}

\makeatother

\usepackage{pgfplots}

\newcommand{\E}{\mathbb{E}} 
\newcommand{\p}{\mathbb{P}} 

\DeclareMathOperator{\mcE}{\mathcal{E}}
\DeclareMathOperator{\mcF}{\mathcal{F}}

\begin{document} 


\title{\sc\bf\large\MakeUppercase{
Can one hear the shape of a population history?}\footnote{Our title is inspired by the famous paper of Mark Kac~\cite{kac1966can}; analogously, we study the theoretical limits to inferring a population size history.}}
\author{
	Junhyong Kim
	\thanks{University of Pennsylvania; \texttt{junhyong@sas.upenn.edu}.}
	\and
	Elchanan Mossel
	\thanks{University of California, Berkeley; \texttt{mossel@stat.berkeley.edu}; supported by NSF grants DMS 1106999 and CCF 1320105 and by DOD ONR grant N000141110140.}
	\and
	Mikl\'os Z. R\'acz
	\thanks{University of California, Berkeley; \texttt{racz@stat.berkeley.edu}; supported by NSF grant DMS 1106999 and by DOD ONR grant N000141110140.}
	\and
	Nathan Ross
	\thanks{University of Melbourne; \texttt{nathan.ross@unimelb.edu.au}.}
}
\date{\today}
\maketitle


\begin{abstract}
Reconstructing past population size from present day genetic data is a major goal of population genetics. Recent empirical studies infer population size history using coalescent-based models applied to a small number of individuals. Here we provide tight bounds on the amount of exact coalescence time data needed to recover the population size history of a single, panmictic population at a certain level of accuracy. In practice, coalescence times
are estimated from sequence data and so our lower bounds should be taken as rather conservative.


\end{abstract}


\section{Introduction}\label{sec:intro} 

Reconstructing the past size and structure of the population of a species is a major goal of population genetics
with applications in, for example, ecology, epidemiology~\cite{heled2008bayesian}, and paleoanthropology~\cite{Li2011}.
It is also important for understanding relationships between different evolutionary parameters, e.g.,
the dynamics of different parts of the genome or how demography affects selection~\cite{Li2012}.

Inference is based on sequence data from individuals sampled from the population under consideration.
Under a given population history, 
the coalescent is a model that provides likelihoods of observed genetic data and is one of the main tools used to infer population history. 
But the space of population histories typically considered is huge and so maximum likelihood estimation requires approximation techniques
\cite{Bhaskar2014, Excoffier2013, harris2013inferring, Li2011, Nielsen2000, palamara2012length, Sheehan2013} which 
lack theoretical guarantees; the same statement applies to Bayesian methods~\cite{Drummond2005, heled2008bayesian}. 
(These methods are discussed in greater detail in Section~\ref{sec:back} below.) 

Here we provide provable information-theoretic lower bounds on the amount of coalescence data needed to estimate, up to some specified accuracy, 
events in a population's past history (see Theorem~\ref{thm:main_bd} below). 
Our bounds are asymptotically tight as shown by analysis of a simple inference algorithm which recovers the history given slightly more data than required by the lower bounds.

Before stating our results in more detail, we provide a brief introduction to inference using the coalescent, as well as a summary of existing literature in this area.

\subsection{Inference using the coalescent}\label{sec:back}

Let $N(t)$ be the size of a single panmictic haploid population at time $t$ ``generations" in the past\footnote{Throughout the paper the unit of time is generations.}
 and call $N = \left\{N\left(t\right) \right\}_{t \geq 0}$ 
the \emph{shape of the population size history}, or simply the \emph{population shape}. 
Given $N$,
Kingman's coalescent (see~\cite{Tavare2004} for background) is a random genealogy on $n$ sampled individuals from the present day population. 
The basic description is that the rate of coalescence between any two individuals/lineages at time $t$ in the past is $1/N(t)$ 
 and so given $k$ lineages at time $t$, the rate of coalescence is $\binom{k}{2}/N(t)$. 
We focus primarily on the case where data comes from pairs of individuals, i.e., $n=2$, just as in, e.g.,~\cite{Li2011}.

The population shape $N = \left\{ N\left(t \right) \right\}_{t \geq 0}$ determines a distribution $\p_{N}$ over coalescent trees;
and in particular, $\p_{N}$ determines the distribution of coalescent trees of any finite number of individuals, at any number of independent loci. 
The first step to infer $N$ using the coalescent is to
ensure that the distribution over coalescent trees uniquely determines the shape of a population history, i.e., that $N \neq N'$ implies that $\p_{N} \neq \p_{N'}$. 
This is indeed true: if we know $\p_N$, then we also know the rate of coalescence of two arbitrary individuals at any time $t$, which is just $1/N(t)$. 
Thus with an infinite amount of coalescence time data, the population shape can be reconstructed.

Considering sequence data, the model assumes that for $n$ individuals in a population, each genomic site follows an $n$-coalescent tree. 
Two sites have the same coalescent tree if there is no recombination breakpoint between them. 
At each site, mutations occur on top of the trees according to a Poisson process with small mutation rate and so, in principle, likelihoods of statistics of sequence data can be derived. 
Unfortunately, recombination is a complicated process and even under simplifying assumptions, likelihood functions are typically intractable, both analytically and computationally. 
Thus inexact methods must be developed, which we now describe.

Given whole genome data, likelihoods of various population parameters can be estimated across the parameter space by MCMC~\cite{Nielsen2000, Excoffier2013}. 
Using a simplified model of recombination~\cite{McVean2005}, simpler likelihood functions arise; however, these must still be analyzed using approximation schemes~\cite{Li2011, Sheehan2013}. 
Sequence data can also be used to infer the lengths of nonrecombinant blocks~\cite{jasmine2014ibd}, the distribution of which can be used to infer various aspects of the population history~\cite{palamara2012length}. 

The problem simplifies when it is assumed that all loci in a given sequence are linked, 
that is, not separated by recombination events or in parts of the genome where no recombination occurs (such as mitochondrial DNA). 
In such cases, the coalescent trees at each of the sites are \emph{identical}. 
For such data, given the coalescent tree, the number of segregating sites (where mutations have occurred) follows a Poisson distribution and analytic (though intractable)
expressions for likelihoods can be derived.

If all sites are unlinked, one can use inference tools involving the population allele frequency spectrum~\cite{Bhaskar2014} or Bayesian approaches such as the ``Bayesian Skyline''~\cite{Drummond2005, heled2008bayesian} for both single and multi-locus data. 
Outside of the coalescent framework, the allele frequency spectrum and its diffusion approximation~\cite{gutenkunst2009inferring, Lukic2011} (which is derived from the underlying Wright-Fisher dynamics that also drive the coalescent) can also be used, though the same computational caveats as above apply. 
The allele frequency spectrum suffers from identifiability issues in general~\cite{myers2008can}, though not under biologically realistic assumptions~\cite{bhaskar2013identifiability}.

\subsection{Results overview and applications}\label{sec:low}

We provide lower bounds on the amount of \emph{exact} coalescence time data necessary to infer past population history events.
The assumption that our data are exact coalescence times is unrealistic but idealized: for a single, panmictic population, the rate of coalescence $t$ generations in the past determines the population size at that time and so the most direct route to estimating the population history is through the coalescence times. Since our lower bounds on 
the amount of samples are for idealized data, the bounds 
should also be taken to apply to methods which use sequence data (and should be considered as underestimates for such methods). 
In fact, all the previously mentioned coalescent-based methods
used to infer population history based on sequence data also infer the coalescence times along the way (usually implicitly).

The following theorem provides bounds on the probability of correctly distinguishing between two population histories that
differ only on an interval $(T, T+S)$ 
over which each is constant, given coalescence times between pairs of individuals at $L$ independent loci.
See Figure~\ref{fig:bottle} for an illustration of two such histories.

\begin{theorem}\label{thm:main_bd}
Let $a, b$, and  $S$ be positive constants and let $T \geq 0$. 
Consider the following hypothesis testing problem: under both hypotheses the population sizes are equal
in the intervals $[0,T)$ and $[T+S,\infty)$, given by some function $N(\cdot)$, 
but under $H_1$ the population size is constant $aN(0)=:aN_0$ in the interval $[T,T+S)$  while 
under $H_2$ the population size during the interval $[T,T+S)$ is constant $bN_0$.
If $L$ independent coalescence times 
are observed from 
either $H_1$ or $H_2$, with prior probability $1/2$,  
then the Bayes error rate for any classifier is at least $(1-\mcE)/2$, where  
\begin{align}
 \mcE^2 &\leq 2L \exp \left( - \int_0^{T} 1 / N\left( t \right) dt \right)\left(1-e^{-\frac{S}{2N_0} \frac{a+b}{ab}}\right)\frac{\left(\sqrt{a}-\sqrt{b} \right)^2}{a+b} \label{eq:main_bd1} \\
	 &\leq  2 L \exp \left( - \int_0^{T} 1 / N\left( t \right) dt \right) \min \left\{ \frac{S}{2N_0}, \frac{ab}{a+b} \right\} \frac{\left( \sqrt{a} - \sqrt{b} \right)^2}{ab}. \label{eq:main_bd2}
\end{align}
In other words, for any classification procedure, 
 the chance of correctly determining  whether the samples came from $H_1$ or $H_2$,  
is at most $(1+\mcE)/2$. 
\end{theorem}

\begin{figure}[h!]
\begin{center}
\begin{tikzpicture}
 \begin{axis}[
 	xlabel=\tiny{Time (generations)}, 
 	ylabel=\tiny{Effective population size}, 
 	xmin=0, 
 	xmax=3, 
 	ymin=0,
 	ymax=4,
 	height=5cm, 
 	width=7cm,
 	y label style={at={(axis description cs:-0.01,0.5)},anchor=north},
 	xticklabels={, , T, T+S, },
 	yticklabels={,, {$b N_0$}, {$a N_0$}, {$N_0$}}
 ]
 \addplot[color=blue,domain=0:1,samples=100]
 {
 4.5*x^3-7*x^2+2*x+3
 };
 \addplot[color=blue] 
 coordinates {
 	(1,2)
 	(2,2)
 	};
 \addplot[color=blue,domain=2:3,samples=100]
 {
 .2*x^2-.1
 };
 \end{axis}
 \end{tikzpicture}
 \hspace{5mm}
 \begin{tikzpicture}
 \begin{axis}[
 	xlabel=\tiny{Time (generations)}, 
 	ylabel=\tiny{Effective population size}, 
 	xmin=0, 
 	xmax=3, 
 	ymin=0,
 	ymax=4,
 	height=5cm, 
 	width=7cm,
 	y label style={at={(axis description cs:-0.01,0.5)},anchor=north},
 	xticklabels={, , T, T+S, },
 	yticklabels={,, {$b N_0$}, {$a N_0$}, {$N_0$}}
 	]
 \addplot[color=blue,domain=0:1,samples=100]
 {
 4.5*x^3-7*x^2+2*x+3
 };
 \addplot[color=blue] 
 coordinates {
 	(1,1)
 	(2,1)
 	};
 \addplot[color=blue,domain=2:3,samples=100]
 {
 .2*x^2-.1
 };	
 \end{axis}
\end{tikzpicture}
\end{center}
\caption{An illustration of two population histories for which Theorem~\ref{thm:main_bd}
provides lower bounds on the amount of coalescence time data needed to distinguish between them.}
\label{fig:bottle}
\end{figure}
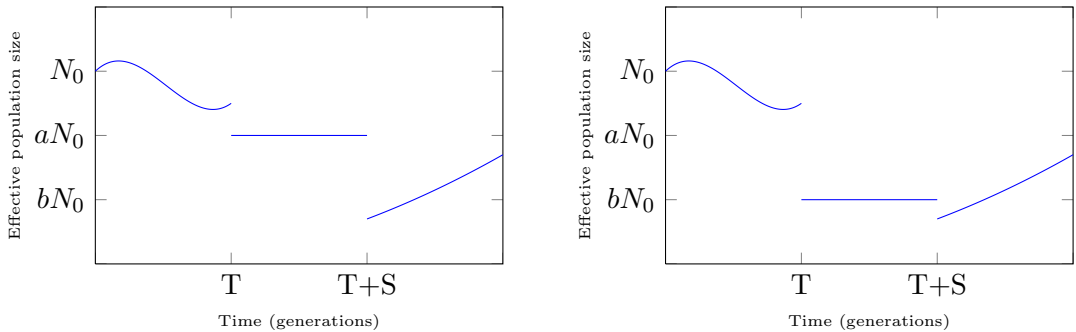

The main features of the bound of the theorem above are that if $L \ll \left( S/N_0 \right)^{-1}$, or if $L\ll 1/(\sqrt{a}-\sqrt{b})^2$, 
or if $T$ is large enough, 
then the chance of distinguishing between the two histories will be near $1/2$. 
Consequently, given $a,b,S$, and $T$, the theorem provides a lower bound on the number $L$ of independent coalescence times necessary in order to distinguish between the two histories with a given probability.

To understand the bound in more concrete settings and to
compare Theorem~\ref{thm:main_bd} to previous work, consider Li and Durbin~\cite{Li2011}, who apply the pairwise sequentially Markovian coalescent model (PSMC) to
the complete diploid genome sequences of seven individuals in order to infer human population size history,
with one of their main goals being to infer the timing of the out-of-Africa event which caused a bottleneck in East Asian and European populations. 
To validate their model, they apply PSMC to simulated data where the population histories consist of a sharp out-of-Africa bottleneck followed by a population expansion. 
They note that the simulations ``reveal a limitation of PSMC in recovering sudden changes in effective population size.'' 
We use Theorem~\ref{thm:main_bd} to quantitatively show that  \emph{every method} must suffer from this to a certain extent. 
%
%

Take the population history considered in~\cite[Fig.~2a]{Li2011}, reproduced in the left panel of Fig.~\ref{fig:LiDur} below. Here the present day effective population size is $N_0:=N(0) = 2.732 \times 10^4$; 
the effective population size is $N_0$ in the time interval $[0,2.732 \times 10^4)$ back in time (measured in years, assuming $25$ years per generation), 
it is $0.05\times N_0$ in the time interval $[2.732 \times 10^4, 1.0245 \times 10^5)$ back in time, 
it is $0.5 \times N_0$ in the time interval $[1.0245 \times 10^5, 3.415 \times 10^6)$ back in time, 
and it is $N_0$ in the time interval $[3.415 \times 10^6, \infty)$ back in time. 
We apply Theorem~\ref{thm:main_bd} to obtain bounds on the amount of data needed to estimate 
the timing of the bottleneck at approximately $100$ kyr to a given accuracy. 
We have $a = 0.05$, $b=0.5$, and $N_0 = 2.732 \times 10^4$. Assuming $25$ years per generation, we have $T = 0.15 N_0$ and $\int_0^{T} 1 / N\left( t \right) dt = \int_0^{0.04 N_0} 1/N_0 dt + \int_{0.04 N_0}^{0.15 N_0} 1/(0.05 N_0) dt = 2.24$. 

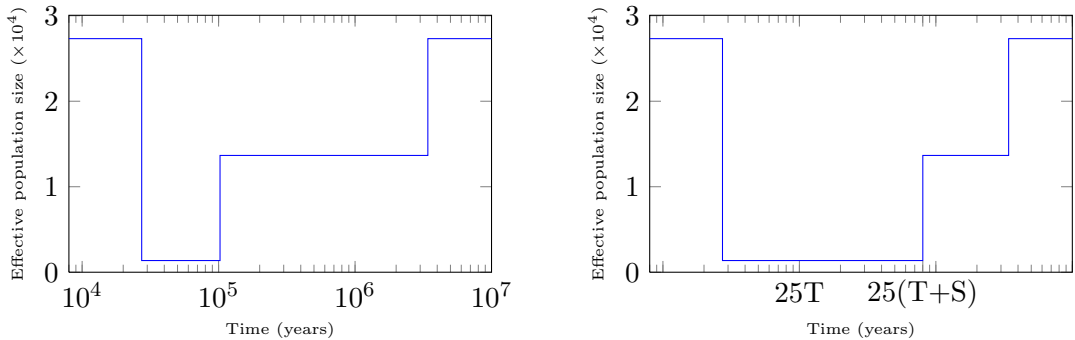
\begin{figure}[h!]
\begin{center}
\begin{tikzpicture}
 \begin{semilogxaxis}[
 	xlabel=\tiny{Time (years)}, 
 	ylabel=\tiny{Effective population size $(\times 10^4)$}, 
 	xmin=8000, 
 	xmax=10000000, 
 	height=5cm, 
 	width=7.2cm,
 	y label style={at={(axis description cs:0.05,0.5)},anchor=north},
 	ymin=0
 ]
 \addplot[color=blue] 
 coordinates {
 	(8000,2.732)
 	(27320,2.732)
 	(27320,2.732/20)
 	(102450,2.732/20)
 	(102450,2.732/2)
 	(3415000,2.732/2)
 	(3415000,2.732)
 	(10^7,2.732)
 };
 \end{semilogxaxis}
 \end{tikzpicture}
 \hspace{5mm}
 \begin{tikzpicture}
 \begin{semilogxaxis}[
 	xlabel=\tiny{Time (years)}, 
 	ylabel=\tiny{Effective population size $(\times 10^4)$}, 
 	xmin=8000, 
 	xmax=10000000, 
 	height=5cm, 
 	width=7.2cm,
 	y label style={at={(axis description cs:0.05,0.5)},anchor=north},
 	ymin=0,
 	xticklabels={, , 25T, , },
 	extra x ticks={802450},
 	extra x tick labels={25(T+S)},
 	extra x tick style={%
      text height=\heightof{0},
   },
 ]
 \addplot[color=blue] 
 coordinates {
 	(8000,2.732)
 	(27320,2.732)
 	(27320,2.732/20)
 	(802450,2.732/20)
 	(802450,2.732/2)
 	(3415000,2.732/2)
 	(3415000,2.732)
 	(10^7,2.732)
 };
 \end{semilogxaxis}
\end{tikzpicture}
\end{center}
\caption{The population histories compared in Table~\ref{table:bottle}. The history on the left is Fig.~2a in~\cite{Li2011} and 
the history on the right is a modified version. For each value of $S$, Theorem~\ref{thm:main_bd} provides lower bounds on the amount of coalescence time 
data needed to distinguish between the two histories.}
\label{fig:LiDur}
\end{figure}

Theorem~\ref{thm:main_bd} tells us that, 
given coalescence times from $L$ independent loci, 
in order to distinguish between the two histories considered in Figure~\ref{fig:LiDur} with probability at least $0.95$, 
it is necessary that $\mcE^2 \geq 0.81$, 
so using~\eqref{eq:main_bd1} and plugging in the numbers above, 
it is necessary that 
\begin{equation}\label{eq:LiDurbinBD}
 2L e^{-2.24} \left( 1 - e^{-\frac{S}{54640} \times \frac{0.55}{0.025}} \right) \frac{\left( \sqrt{0.5} - \sqrt{0.05} \right)^2}{0.55} \geq 0.81.
\end{equation}
From this we immediately see that there is no solution for $S$ when $L \leq 8$, 
i.e., when the number of independent loci is too small, 
the  length of the 95\% ``uncertainty interval''\footnote{For a statistical hypothesis $H_1: \theta=\theta_1$ on a real parameter $\theta$ and a given set of data,
we define the $95\%$ \emph{uncertainty interval} to be the set of
values $\theta_2$ such that  for any classification procedure, 
 the chance of correctly determining  whether the samples came from $H_1$ or $H_2: \theta=\theta_2$ is no greater than  $0.95$.} is infinite.
When $L \geq 9$,~\eqref{eq:LiDurbinBD} is equivalent to
\begin{equation}\label{eq:LiDurbinBD2}
 S \geq \frac{27320}{11}  \log \left( 1+ \frac{\frac{891 \times e^{2.24}}{4000 \times \left( \sqrt{0.5} - \sqrt{0.05} \right)^2}}{L - \frac{891 \times e^{2.24}}{4000 \times \left( \sqrt{0.5} - \sqrt{0.05} \right)^2}} \right),
\end{equation}
i.e., the length of the 95\% ``uncertainty interval'' is inversely proportional to $L$ for large $L$. 
Table~\ref{table:bottle} collects the numerical values of some of the lower bounds on the lengths of the 95\% uncertainty intervals given by~\eqref{eq:LiDurbinBD} and~\eqref{eq:LiDurbinBD2} 
(note that in~\eqref{eq:LiDurbinBD} and~\eqref{eq:LiDurbinBD2} the unit of $S$ is generations, while in Table~\ref{table:bottle} the unit of time is years, where we assume $25$ years per generation). 
These estimates are in line with the simulation results of~\cite{Li2011}, where in the PSMC reconstruction of the population history the sudden drop in population is spread out over several tens of thousands of years.

\begin{table}[h]
\begin{center}
\begin{tabular}{|  l  || c c c c c  |}
\hline 
 Number of loci & $\leq 8$ & 10 & 20 & 30 & 50 \\
 \hline
Lower bound on interval length (in years) &  $\infty$ & $1.3 \times 10^5$  & $3.6 \times 10^4$ & $2.1 \times 10^4$ & $1.2 \times 10^4$ \\
\hline
\end{tabular}
\caption{Lower bounds on the lengths of the 95\% uncertainty intervals for determining the timing of a bottleneck given a sample of $L$ independent loci in the scenario depicted in
Figure~\ref{fig:LiDur}.}
\label{table:bottle}
\end{center}
\end{table}


Similarly, Theorem~\ref{thm:main_bd} also provides bounds on the amount of coalescence data needed to estimate 
the time of the final jump in the population to $N_0$ in this same scenario. 
We may consider two population histories, one the same as Fig.~2a in~\cite{Li2011}, 
and the other a modified version where the final jump in the population to $N_0$ occurs at some time in the interval  $[1.0245\times 10^5,  3.415\times 10^6]$ years in the past; 
see Fig.~\ref{fig:LiDur2}.

\begin{figure}[h!]
\begin{center}
\begin{tikzpicture}
 \begin{semilogxaxis}[xlabel=\tiny{Time (years)}, ylabel=\tiny{Effective Population size $(\times 10^4)$}, xmin=8000, xmax=10000000, height=5cm, width=7.2cm,
 	y label style={at={(axis description cs:0.05,0.5)},anchor=north}, ymin=0,
 	xticklabels={ ,$10^4$ , $10^5$, ,  },
 	extra x ticks={3415000},
 	extra x tick labels={$25 (T+S)$},
 ]
 \addplot[color=blue] 
 coordinates {
 	(8000,2.732)
 	(27320,2.732)
 	(27320,2.732/20)
 	(102450,2.732/20)
 	(102450,2.732/2)
 	(3415000,2.732/2)
 	(3415000,2.732)
 	(10^7,2.732)
 };
 \end{semilogxaxis}
 \end{tikzpicture}
 \hspace{5mm}
 \begin{tikzpicture}
 \begin{semilogxaxis}[xlabel=\tiny{Time (years)}, ylabel=\tiny{Effective Population size $(\times 10^4)$}, xmin=8000, xmax=10000000, height=5cm, width=7.2cm,
 	y label style={at={(axis description cs:0.05,0.5)},anchor=north,}, ymin=0,
 	xticklabels={, $10^4$ , $10^5$, ,  },
 	extra x ticks={602450, 3415000},
 	extra x tick labels={25$T$, 25$(T+S)$},
 	extra x tick style={%
   },
 ]
 \addplot[color=blue] 
 coordinates {
 	(8000,2.732)
 	(27320,2.732)
 	(27320,2.732/20)
 	(102450,2.732/20)
 	(102450,2.732/2)
 	(602450,2.732/2)
 	(602450,2.732)
 	(10^7,2.732)
 };
 \end{semilogxaxis}
\end{tikzpicture}
\end{center}
\caption{The population histories compared in Table~\ref{table:bd4}. The history on the left is Fig.~2a in~\cite{Li2011} and 
the history on the right is a modified version. 
Here $25(T+S)=3.415\times10^6$ and for each value of $T$, Theorem~\ref{thm:main_bd} provides lower bounds on the amount of coalescence
data needed to distinguish between the two histories.}
\label{fig:LiDur2}
\end{figure}

Here we thus have $a=0.5$, $b=1$, $25(T+S)=3.415\times 10^6$, and $25T \in [1.0245\times 10^5,  3.415\times 10^6]$. For each such $T$ we have 
\begin{align*}
\int_0^T 1/N(t) dt & =\int_0^{.04 N_0}  1/N_0 dt +\int_{.04 N_0}^{.15 N_0}  1/(.05 N_0) dt + \int_{.15N_0}^{T} 1/(. 5 N_0) dt =1.94 + 2(T/N_0).
\end{align*}
Plugging these expressions 
into~\eqref{eq:main_bd2} of Theorem~\ref{thm:main_bd} and taking $ab/(a+b)$ in the minimum,
we find that, in order to recover the true population history with probability at least $0.95$ given $L$ independent samples, 
we must have
\begin{equation}\label{eq:lastjump}
 T \leq 13660 \log \left( L / 49.2 \right).
\end{equation}
Notice again that recovery with $95\%$ chance is impossible when $L \leq 49$, so a considerable amount of coalescence time data is required for accurate inference. 
When $L \geq 50$, Table~\ref{table:bd4} summarizes lower bounds on the length of the $95\%$ uncertainty intervals implied by~\eqref{eq:lastjump} (again converted to years).

\begin{table}[h]
\begin{center}
\begin{tabular}{|  l  || c c c c c  |}
\hline 
 Number of loci & $100$ & $200$ & $500$ & $10^3$ & $10^4$ \\
 \hline
Lower bound on interval length (in yr) &  $3.1 \times 10^6$ & $2.9 \times 10^6$ & $2.6 \times 10^6$ & $2.3 \times 10^6$ & $1.6 \times 10^6$ \\
\hline
\end{tabular}
\caption{Lower bounds on the lengths of the 95\% uncertainty intervals for determining the timing of the last population size change in the scenario depicted in
Figure~\ref{fig:LiDur2}, given a sample of $L$ independent loci.}
\label{table:bd4}
\end{center}
\end{table}

The PSMC reconstruction of the population history ends at approximately 5 Myr, so the lengths of their uncertainty intervals on the timing of the last population size change are unclear, but appear to be at least a few Myr. Our results are therefore in line with these simulation results, and show that no method can perform substantially better than PSMC.

\subsection{Organization of the paper}
The layout of the paper is as follows:
in the next section we describe a procedure that infers a population history given slightly more
data than required by the lower bounds implied by Theorem~\ref{thm:main_bd}. We prove our results on lower bounds, including Theorem~\ref{thm:main_bd}, in
Section~\ref{sec:lower}, and then prove results about the inference procedure in Section~\ref{sec:algo}. 
We support our results by simulations presented in Section~\ref{sec:sim}, and we end with a summarizing discussion section with some open problems. 

\section{A Simple Reconstruction Algorithm}\label{sec:results} 

To complement the results on lower bounds detailed above, 
we describe a simple estimation procedure (analyzed in Section~\ref{sec:algo}) 
that takes coalescence time data and
returns an estimate for the population shape. 
The analysis of this procedure shows that the amount of data it requires almost matches the lower bounds stated in our results above. 

The procedure takes i.i.d.\ pairwise coalescence times $\mathbf{t}^L = \left\{ t_1, \dots, t_L \right\}$
and returns a piecewise constant estimate $\widehat{N} = \left\{ \widehat{N} \left( t \right) \right\}_{t \geq 0}$ for the population shape. 
The procedure involves a single parameter, $\eps$, which controls the length of the time intervals where our estimate is constant, and which then also affects the accuracy of our estimate in each time interval. 
Assume that there are $N_0$ individuals initially, at time $0$. 
\begin{enumerate}
 \item Partition time backwards in time into intervals of length $\eps N_0$, i.e., let $I_1 = \left[ 0, \eps N_0 \right]$, $I_2 = \left[ \eps N_0, 2 \eps N_0 \right], \dots, I_K = \left[ \left( K - 1 \right) \eps N_0, K \eps N_0 \right]$. 
 ($K$ is the minimum integer such that the interval $[0, K \eps N_0]$ covers the data and we do not provide estimates past time $K \eps N_0$.)
 \item For $k = 1, \dots, K$, denote the fraction of data points lying in the time interval $I_k$ by
\[
 \widehat{X}_k := \frac{1}{L} \# \left\{ i : t_i \in I_k \right\},
\]
and furthermore let $\widehat{S}_0 = 0$ and $\widehat{S}_k = \sum_{i=1}^{k} \widehat{X}_i$, the fraction of data points lying in the time interval $\left[ 0, k \eps N_0 \right]$.
 \item Our estimate $\widehat{N}_k$ in the time interval $I_k$ is
\begin{equation}\label{eq:est}
 \widehat{N}_k := \frac{\eps N_0}{- \log \left( 1 - \frac{\widehat{X}_k}{1-\widehat{S}_{k-1}} \right)},
\end{equation}
provided that $\widehat{X}_k > 0$, i.e., we have at least one data point in the time interval $I_k$. If $\widehat{X}_k = 0$, then we do not give an estimate.
\end{enumerate}

\begin{remark}
The estimate~\eqref{eq:est} is motivated by the fact that
\[
 \frac{\p \left( t_1 \in I_k \right)}{\p \left( t_1 \notin \left[ 0, \left( k - 1 \right) \eps N_0 \right] \right)} = 1 - \exp \left( - \int_{\left( k -1 \right) \eps N_0}^{k\eps N_0} \frac{1}{N \left( t \right)} dt \right).
\]
\end{remark}
\begin{remark}

In Step~1 above, we partition time into intervals of equal length. 
This is done solely to make the subsequent analysis and discussion as simple as possible. 
Depending on the specific application, it might be of interest to consider other choices of partitions, for instance, choosing intervals whose lengths grow exponentially backwards in time. 
Our estimation procedure (and also the subsequent analysis) works in an analogous way: $\widehat{X}_k$ and $\widehat{S}_k$ can be defined in the same way in Step~2, and the only change in the estimate~\eqref{eq:est} is to replace $\eps N_0$ in the numerator of the fraction with the length of the appropriate interval, $\left| I_k \right|$.

\end{remark}
In order to state the properties of this procedure, define for $k \geq 1$ the  ``effective constant population size in the time interval $I_k$'' by 
\[
 \widetilde{N}_k := \frac{\eps N_0}{\int_{\left( k - 1 \right) \eps N_0}^{k\eps N_0} \frac{1}{N\left( t \right)} dt};
\]
the $\widetilde{N}_k $ give a natural piecewise constant approximation of the population shape $N(t)$ that is directly
comparable to the piecewise estimate  $\widehat{N}$.
Let 
\[
 E_k := \sup_{t \in I_k} \left| \log N\left( t \right) - \log \widehat{N}_k \right|
\]
be the absolute error of our estimate $\widehat{N}$ on a logarithmic scale for each time interval. 
When estimating the error $E_k$, there are two types of errors to consider. 
One is the inherent error coming from the fact that we are approximating the shape with a piecewise constant function; the other error comes from the finite sample size $L$. 
By the triangle inequality we can bound the error $E_k$ by the sum of these two errors:
\[
 E_k \leq E_{k,1} + E_{k,2},
\]
where
\[
 E_{k,1} := \sup_{t \in I_k} \left| \log N \left( t \right) - \log \widetilde{N}_k \right|
\]
is the error coming from approximating the shape in the time interval $I_k$ with a constant, and
\[
 E_{k,2} := \left| \log \widetilde{N}_k - \log \widehat{N}_k \right|
\]
is the error coming from the finite sample size. Ignoring the error $E_{k,1}$ for now, 
we can use concentration inequalities to derive the following finite sample estimate for the accuracy of our estimator: 
\begin{proposition}\label{prop:k_cond_intro}
Given that $\ell$ samples ``survived'' the first $k-1$ intervals, the probability that $\log \widetilde{N}_k$, is in the (random) interval
\begin{equation*}\label{eq:k_cond_int_intro}
 \left[ \log \left( \eps N_0 \right) - \log \left( - \log \left( \left( 1 - \frac{L}{\ell}\widehat{X}_k - c \right) \vee 0 \right) \right), \log \left( \eps N_0 \right) - \log \left( - \log \left( \left( 1 - \frac{L}{\ell}\widehat{X}_k + c \right) \wedge 1 \right) \right) \right]
\end{equation*}
is at least $1 - 2 \exp \left( - 2 c^2 \ell \right)$ for all $c \geq 0$.\footnote{Here and in the following we use the notation $a\vee b = \max \left\{ a, b \right\}$ and $a \wedge b = \min \left\{ a, b \right\}$.} 
\end{proposition}

Note that the interval in the proposition contains the estimate $\log \widehat{N}_k$.

To understand in what sense Proposition~\ref{prop:k_cond_intro} and Theorem~\ref{thm:main_bd} are matching bounds, first
consider the following easy corollary of Theorem~\ref{thm:main_bd} that better matches the setting of Proposition~\ref{prop:k_cond_intro}.

\begin{theorem}\label{lem:bd4_intro}
Let $a, b$, and $S$ be positive constants and $T \geq 0$. Consider the following hypothesis testing problem: 
Under $H_1$ the population size is constant $aN(0)=:aN_0$ in the interval $[T,T+S)$  while 
under $H_2$ the population size during the interval $[T,T+S)$ is constant $bN_0$.
%
Assume the data are independent coalescence times  and let
$\ell$ be the number of  pairs that have not coalesced by time $T$. 
If the true history is given by either $H_1$ or $H_2$, each with prior probability $1/2$, 
then the Bayes error rate for any classifier is at least $(1 - \Delta)/2$, 
where $\Delta$ satisfies: 
\begin{equation*}\label{eq:bd4_intro}
 \Delta^2 \leq2\ell \frac{\left( \sqrt{b}-\sqrt{a}\right)^2}{a+b}. 
\end{equation*}
In other words, for any classification procedure, the chance of correctly determining whether the samples came from $H_1$ or $H_2$, is at most $(1+\Delta)/2$. 
\end{theorem}

Writing $b=a(1+\eta)$ for $\eta>0$, the bound of the theorem becomes $2\ell (1-2\sqrt{1+\eta}/(2+\eta))\leq \ell \eta^2/4$
and so we see that if $\eta \ll \ell^{-1/2}$, then \emph{no procedure} will distinguish between the two histories  given
by $H_1$ and $H_2$ with good probability.
On the other hand, Proposition~\ref{prop:k_cond_intro} implies that
for a fixed confidence $\alpha$,
\begin{equation*}
1 - 2 \exp \left( - 2 c^2 \ell \right)=\alpha,
\end{equation*}
and the constant $c$ is of order $\ell^{-1/2}$ as $\ell$ becomes large, and thus the width of the interval
in Proposition~\ref{prop:k_cond_intro} is of order $\log(1+C \ell^{-1/2})$ where $C$ is some constant. 
To summarize, if $\eta \gg \ell^{-1/2}$, then our method will distinguish between the histories with high probability; but if
$\eta \ll \ell^{-1/2}$, Theorem~\ref{lem:bd4_intro} shows that \emph{no procedure} will distinguish between the two histories with good probability.
On a conceptual level, this last statement is the main purpose of the paper:
a significant amount of data is needed to infer past population size,
especially in deep history where there is likely to be little coalescence information.

%

For illustration, we implement our estimation procedure on simulated data in Section~\ref{sec:sim}, where we find a good general performance, matching our theoretical results.

%

\subsection{Related theoretical work}\label{sec:related} 

Our reconstruction algorithm is a special case of the following problem: given $n$ i.i.d.\ copies of the \emph{first point} of a Poisson point process on $[0,\infty)$ with intensity $\phi(t)$,
what is a good estimate of $\phi(t)^{-1}$?
Poisson process intensity estimation has a large literature, see for example \cite{Birge2007, Reynaud-Bouret2003, Willett2007} and references therein, but the (natural) data assumed in this area is one realization of the point process,
or the point process observed up to some fixed time, or i.i.d.\ copies of such data, which does not fit our framework.

For another perspective to this question, define the hazard rate for a positive random variable $X$ with density $f$ and distribution function $F$ to be
\begin{equation}
-\frac{d}{dt} \log(1-F(t))=f(t)/(1-F(t)). \label{222}
\end{equation}
A simple calculation shows that the time of the first point of a Poisson process with intensity $\phi(t)$ has the same distribution as a positive random variable 
with hazard rate $\phi(t)$. Due largely to their importance in applications in, e.g., insurance, medicine, and reliability theory \cite[Section 1.1]{Lawless2003}, hazard rate estimation is well studied;
some seminal papers are~\cite{Rice1976, Sethuraman1981, Yandell1983} and see the recent~\cite{Cheng2006} and references there. 
Without embellishments specific to lifetime data (such as censoring where some lifetimes are only known to be at least some value), the main technique to estimating \eqref{222} (which also applies to its inverse) is 
to adapt estimators of $f$ and $F$. 

Indeed, our reconstruction algorithm is essentially an adaptation of the histogram estimate of the density and distribution function to our setting. 
Other popular density estimation techniques such as those in the introduction of~\cite{Silverman1986} can be adapted to our setting through the use of~\eqref{222}; for example
see~\cite{Wang2005} for a survey of kernel smoothing methods for hazard function estimation.
Our particular estimation procedure was chosen due to its simplicity and explicitness; in particular, we mention two points. 
The first is that we desire results like Proposition~\ref{prop:k_cond_intro}  with explicit non-asymptotic confidence intervals. Asymptotic confidence intervals can be 
obtained and used as estimates for smoothed density estimators, but with error depending on unknown quantities related to the underlying density which 
can lead to poor coverage accuracy~\cite{Hall1992}.
Secondly, smoothed density estimators have improved performance only when the underlying density is itself smooth (expressed as differentiability and continuity conditions). A major
purpose of estimating past population size is to discover drastic changes in population size such as bottlenecks~\cite{harris2013inferring,Li2011,palamara2012length,Sheehan2013},
when it is not clear such smoothness assumptions are appropriate. 
\section{Proof of Lower Bounds}\label{sec:lower} 

In this section we prove Theorem~\ref{thm:main_bd}, as well as
derive some other lower bounds for the amount of data needed for a given accuracy of estimating the population shape. 
This is done by formulating hypothesis tests deciding between two population shapes, and proving upper bounds on the probability of correctly inferring the population shape.

\subsection{Background on probability metrics}\label{sec:prob_metrics} 

We first recall a few metrics between probability distributions (see~\cite{gibbs2002choosing} for a survey). 
Let $P$ and $Q$ be two probability measures that are absolutely continuous with respect to a third probability measure $\lambda$. 
Write $f_P = \frac{dP}{d\lambda}$ and $f_Q = \frac{dQ}{d\lambda}$ for the respective Radon-Nikodym derivatives. 
The square of the Hellinger distance between $P$ and $Q$ is then defined as
\[
 d_H^2 \left( P, Q \right) := \frac{1}{2} \int \left( \sqrt{f_P} - \sqrt{f_Q} \right)^2 d\lambda.
\]
The definition does not depend on the choice of $\lambda$. A nice property of the Hellinger distance is that for product measures $P = P_1 \times P_2$, $Q = Q_1 \times Q_2$, we have that
\[
 1 - d_H^2 \left( P, Q \right) = \left( 1 - d_H^2 \left( P_1, Q_1 \right) \right) \left( 1 - d_H^2 \left( P_2, Q_2 \right) \right),
\]
which immediately implies that
\[
 d_H^2 \left( P, Q \right) \leq d_H^2 \left( P_1, Q_1 \right) + d_H^2 \left( P_2, Q_2 \right).
\]
Another commonly used metric is the total variation distance:
\[
 d_{TV} \left( P,Q \right) := \sup_{A \in \mcF} \left| P \left( A \right) - Q \left( A \right) \right|,
\]
or, equivalently:
\[
 d_{TV} \left( P,Q \right) = \frac{1}{2} \int \left| f_P - f_Q \right| d\lambda.
\]
We use the following well-known fact:
\begin{lemma} With the notation above we have
 \[
  d_{TV} \leq \sqrt{2} d_H.
 \]
\end{lemma}
\begin{proof}
 This follows from the identity $f_P - f_Q = \left( \sqrt{f_P} - \sqrt{f_Q} \right) \left( \sqrt{f_P} + \sqrt{f_Q} \right)$, the Cauchy-Schwarz inequality, and the inequality $\left( \sqrt{f_P} + \sqrt{f_Q} \right)^2 \leq 2 \left( f_P + f_Q \right)$.
\end{proof}

\subsection{A lower bound on the amount of data needed to recover a constant  history}\label{sec:con_hist}

We start in a simpler setting than Theorem~\ref{thm:main_bd} where 
we are trying to differentiate with good probability between two populations of \emph{constant} size.
For this simple setup we assume our data 
are $L$ i.i.d.\ copies of coalescent trees on $n$ individuals from a constant population, and we want to 
estimate the size of the population. 
We derive lower bounds
on the amount of data needed for recovery.

\begin{theorem}\label{lem:bdm_intro}
Consider the following hypothesis testing problem: 
$H_1$ states that the population size during the interval $[0 ,\infty)$ is constant $N$, while 
$H_2$ states that the population size during the interval $[0,\infty)$ is the constant $(1+\eta)N$, where $\eta > 0$ is fixed.
If $L$ i.i.d.\ coalescent trees on $n$ individuals are observed from
either $H_1$ or $H_2$, each with prior probability $1/2$, 
then the Bayes error rate for any classifier is at least $(1-\Upsilon)/2$, 
where $\Upsilon$ satisfies: 
\begin{equation*}
\Upsilon^2 \leq 2L \left(1-\left(\frac{2\sqrt{1+\eta}}{2+\eta}\right)^{n-1}\right) \leq  \frac{L(n-1)\eta^2}{4}.
\end{equation*}
In other words, for any classification procedure, the chance of correctly determining whether the samples come from $H_1$ or $H_2$, is at most  $(1+\Upsilon)/2$.
\end{theorem}

The interpretation of the theorem is that if $\eta \ll (n L)^{-1/2}$, then \emph{no procedure} will distinguish between the two histories  given
by $H_1$ and $H_2$ with good probability. 
In other words, we need $L = \Omega \left( 1 / \left( n  \eta^2 \right) \right)$ samples to differentiate between the two histories $N_1$ and $N_2$.\footnote{We use the standard asymptotic notation $\Omega$, which means ``at least on the order of''. Formally, if $a_n$ and $b_n$ are two sequences such that there exists a positive constant $c$ and an integer $n_0$ such that for every $n \geq n_0$, $a_n \geq c \times b_n$, then $a_n = \Omega \left( b_n \right)$ as $n \to \infty$. 
Similarly, if $f$ and $g$ are two functions such that there exist positive constants $c$ and $x_0$ such that for every $x \in \left( 0, x_0 \right)$, $f \left( x \right) \geq c \times g \left( x \right)$, then $f\left( x \right) = \Omega \left( g \left( x \right) \right)$ as $x \searrow 0$. 
Equivalently, $f\left( x \right) = \Omega \left( g \left( x \right) \right)$ if and only if $g\left( x \right) = O \left( f \left( x \right) \right)$.}
We reiterate that these bounds hold knowing exact rather than estimated coalescence times,
and so should be considered as underestimates in more realistic data settings.  

We set up for the proof of Theorem~\ref{lem:bdm_intro}; the same paradigm will be used to prove Theorem~\ref{thm:main_bd}. 
Consider the following hypothesis testing problem. 
Let $\eta > 0$, and let $N_1 \left( \cdot \right) \equiv N$ and $N_2 \left( \cdot \right) \equiv \left( 1 + \eta \right) N$ be two population size histories. 
Let $\kappa$ be uniform in $\left\{ 1, 2 \right\}$, and, given $\kappa$, let $\mathbf{R}^{\kappa, L} = \left\{ R_1^\kappa, \dots, R_L^\kappa \right\}$ be a collection of
$L$ i.i.d.\ coalescence trees on $n$ individuals drawn from the distribution induced by the population size history $N_{\kappa}$. The problem is to infer $\kappa$ from $\mathbf{R}^{\kappa,L}$.

The probability of correctly inferring $\kappa$ using the optimal reconstruction strategy is clearly at least $1/2$; denote this probability by $\left( 1 + \Upsilon \right) / 2$ (here $\Upsilon = \Upsilon \left( L,n,N, \eta \right)$).  
The reconstruction method which gives the largest probability of correctly inferring $\kappa$ is maximum likelihood: let $\widehat{\kappa} = 1$ if $\p \left( \kappa = 1 \, \middle| \, \mathbf{R}^{\kappa, L} \right) \geq \p \left( \kappa = 2 \, \middle| \, \mathbf{R}^{\kappa, L} \right)$ and $\widehat{\kappa} = 2$ otherwise. Then we have
\[
 \Upsilon = \p \left( \widehat{\kappa} = \kappa \right) - \p \left( \widehat{\kappa} \neq \kappa \right) = d_{TV} \left( \mathbf{R}^{1, L}, \mathbf{R}^{2, L} \right).
\]

\begin{proof}[Proof of Theorem~\ref{lem:bdm_intro}]
By the facts in Section~\ref{sec:prob_metrics} we have
\begin{equation}\label{2201}
 \Upsilon(L,n,N,\eta)^2 = d_{TV}^2\left( \mathbf{R}^{1,L}, \mathbf{R}^{2,L}\right)\leq 2 d_{H}^2\left( \mathbf{R}^{1,L}, \mathbf{R}^{2,L}\right)\leq  2 L d_H^2 \left( R_1^1, R_1^2 \right).
\end{equation}
Since the increasing sequence of times of coalescence of the trees $R_1^i$, denoted by $\mathbf{s}^i=(s_1^i, \ldots, s_{n-1}^i)$,  are sufficient statistics 
for $R_1^i$, we have
$d_H^2 \left( R_1^1, R_1^2 \right)=d_H^2(\mathbf{s}^1,\mathbf{s}^2)$.  We can directly compute the density $f_i(\mathbf{x})$ of $\mathbf{s}^i$ as
\[
f_i(\mathbf{x})=\prod_{j=1}^{n-1} \exp\left\{-\binom{n-j+1}{2}(x_j-x_{j-1})/N_i \right\}\frac{\binom{n-j+1}{2}}{N_i} \mathbf{1}_{\left\{ 0 < x_1 < \dots < x_{n-1} \right\}},
\]
where we have set $N_1:=N$, $N_2:=(1+\eta)N$, and $x_0=0$.
Using these densities in the definition of the Hellinger distance and noting especially that since $f_i$ is a density, 
we have for any $\alpha>0$ that
\[
\int_{0<x_1<\ldots< x_{n-1}} \prod_{j=1}^{n-1} \exp\left\{-\alpha \binom{n-j+1}{2}(x_j-x_{j-1}) \right\}\binom{n-j+1}{2}d\mathbf{x}=1/\alpha^{n-1},
\]
a calculation shows that
\begin{equation}\label{eq:dH1}
 d_H^2 \left( s_1^1, s_1^2 \right) = 1 -  \left(\frac{2 \sqrt{1+\eta}}{2 + \eta}\right)^{n-1}.
\end{equation}
Plugging this into~\eqref{2201} yields the first bound of the result. When $\eta  > 0$, we can upper bound the right hand side of~\eqref{eq:dH1} by $(n-1)\eta^2 / 8$
 to get the simpler bound $\Upsilon^2 \leq L (n-1) \eta^2 / 4$.
\end{proof}

\subsection{Proof of Theorem~\ref{thm:main_bd}}\label{sec:main_bd} 

We prove Theorem~\ref{thm:main_bd} using the same strategy as that of Section~\ref{sec:con_hist}.
For $i=1,2$, let $N_i(\cdot)$ be the history corresponding to hypothesis $H_i$. 
Let $\kappa$ be uniform in $\left\{ 1, 2 \right\}$, and, given $\kappa$, let $\mathbf{t}^{\kappa, L} = \left\{ t_1^\kappa, \dots, t_L^\kappa \right\}$ be a collection of
$L$ i.i.d.\ coalescence times of pairs of individuals
drawn from the distribution induced by the population size history $N_{\kappa}$. The problem is to infer $\kappa$ from $\mathbf{t}^{\kappa,L}$.

\begin{proof}[Proof of Theorem~\ref{thm:main_bd}]
As above, the chance that we infer $\kappa$ correctly from $\mathbf{t}^{\kappa,L}$ is bounded above by 
$(1+\mcE(L,a,b,T,S))/2$ where 
\begin{equation}
\mcE(L,a,b,T,S)^2=d_{TV}^2\left( \mathbf{t}^{1,L}, \mathbf{t}^{2,L}\right)\leq  2 L d_H^2 \left( t_1^1, t_1^2 \right). \label{2200}
\end{equation}
Writing $a_1:=a$ and $a_2:=b$ to shorten formulas, a straightforward calculation shows that the density of $t_1^i$ is 
\[
f_i(x)=
\begin{cases}
 \exp\left(-\int_0^x\frac{1}{N(s)} ds\right) \frac{1}{N(x)},  & x < T,\\
  \exp\left(-\int_0^T\frac{1}{N(s)} ds\right) \exp\left(- \frac{x-T}{a_i N_0}\right) \frac{1}{a_i N_0}, & T\leq x<T+S,\\
  \exp\left(-\int_0^T\frac{1}{N(s)} ds\right)\exp\left(- \frac{S}{a_i N_0}\right)   \exp\left(-\int_{T+S}^x\frac{1}{N(s)} ds\right) \frac{1}{N(x)}, & T+S \leq x.
\end{cases}
\]
Using these densities in the definition of the Hellinger distance, we find after some simple calculations that
\[
d_H^2 \left( t_1^1, t_1^2 \right) = 
 \exp \left( - \int_0^{T} 1 / N\left( t \right) dt \right)\left(1-e^{-\frac{S}{2N_0} \frac{a+b}{ab}}\right)\frac{\left(\sqrt{a}-\sqrt{b} \right)^2}{a+b}.
 \]
Finally, using the inequality $1 - e^{-x} \leq \min \left\{x, 1 \right\}$, simplifying, and plugging the result into~\eqref{2200} implies the bound of the theorem.
\end{proof}

\section{Estimating the population shape}\label{sec:algo} 

Recall our setting of the estimation procedure for the population shape $N = \left\{ N\left(t \right) \right\}_{t \geq 0}$ from the i.i.d.\ coalescence times $\mathbf{t}^L = \left\{ t_1, \dots, t_L \right\}$.
In this section we analyze our piecewise constant population shape estimator $\widehat{N} = \left\{ \widehat{N} \left( t \right) \right\}_{t \geq 0}$ introduced in Section~\ref{sec:results}. 

We consider the absolute error of our estimate $\widehat{N}$ on a logarithmic scale for each time interval, i.e., for $k\geq 1$ we consider
\[
 E_k := \sup_{t \in I_k} \left| \log N\left( t \right) - \log \widehat{N}_k \right|.
\]
Recall also that $E_k \leq E_{k,1} + E_{k,2}$, where $E_{k,1}$ and $E_{k,2}$ are also both defined in Section~\ref{sec:results}.

To bound the error $E_{k,1}$ it is necessary to make an additional assumption on the population shape. 
We introduce an additional parameter, $\delta$, which controls how much the population size can vary within a time interval, and we make the following assumption.

\begin{assumption}\label{ass:delta}
 We assume that in each time interval the population size can increase by a factor of at most $e^{\delta \eps}$, and can decrease by a factor of at most $e^{- \delta \eps}$. 
\end{assumption}

Using this assumption, it is simple to bound the first type of error.
\begin{lemma}
 Given Assumption~1, we have that $E_{k,1} \leq 2 \delta \eps$.
\end{lemma}
\begin{proof}
 Assumption~1 implies that 
\[
 \log N \left( \left( k - 1 \right) \eps N_0 \right) - \delta \eps \leq \log N \left( t \right) \leq  \log N \left( \left( k - 1 \right) \eps N_0 \right) + \delta \eps
\]
for all $t \in I_k$, and consequently also that
\[
 \log N \left( \left( k - 1 \right) \eps N_0 \right) - \delta \eps \leq \log \widetilde{N}_k \leq  \log N \left( \left( k - 1 \right) \eps N_0 \right) + \delta \eps.
\]
These inequalities then imply that $E_{k,1} \leq 2 \delta \eps$.
\end{proof}

To estimate the second type of error, $E_{k,2}$, it is not necessary to make any assumptions. We use concentration results for sums of i.i.d.\ random variables, and, in particular, we use the following simple corollary of the Chernoff bound.
\begin{theorem}\label{thm:Chernoff}
 Let $Y_1, \dots, Y_n$ be i.i.d.\ Bernoulli($p$) random variables, and let $Y = \sum_{i=1}^n Y_i$. Then for any $\lambda > 0$ we have
\begin{equation}\label{eq:Ch1}
 \left.
\begin{aligned}
 &\p \left( Y \leq np - \lambda \right)\\
 &\p \left( Y \geq np + \lambda \right)
\end{aligned}
 \right\}
\leq \exp \left( - \frac{2 \lambda^2}{n} \right).
\end{equation}
\end{theorem}
The bounds in Theorem~\ref{thm:Chernoff} imply the following concentration bound.
\begin{corollary}\label{cor:chernoff}
 For any $k \geq 1$ and $\lambda > 0$ we have
\begin{equation}\label{eq:conc2}
 \p \left( \left| \widehat{X}_k - \E \left( \widehat{X}_k \right) \right| \geq \lambda \right) \leq 2 \exp \left( - 2 \lambda^2 L \right).
\end{equation}
\end{corollary}
In the following we present two bounds on the error $E_{k,2}$. We first present a bound for the first interval (i.e., when $k = 1$), which then also implies conditional bounds for general intervals, by conditioning on the number of data points that have not coalesced by a given time.

\subsection{Bounds for the first interval}\label{sec:bds1} 

\begin{proposition}\label{prop:k1-simple}
 For any $c \geq 0$, with probability at least $1 - 2 \exp \left( - 2 c^2 L \right)$, the logarithm of the effective constant population size in $I_1$, $\log \widetilde{N}_1$, is in the interval
\begin{equation}\label{eq:k1_conf_int}
 \left[ \log \left( \eps N_0 \right) - \log \left( - \log \left( \left( 1 - \widehat{X}_1 - c \right) \vee 0 \right) \right), \log \left( \eps N_0 \right) - \log \left( - \log \left( \left( 1 - \widehat{X}_1 + c \right) \wedge 1 \right) \right) \right].
\end{equation}
\end{proposition}
Note that the interval in~\eqref{eq:k1_conf_int} is an interval around our estimate $\log \widehat{N}_1$.
\begin{proof}
 The inequality~\eqref{eq:conc2} for $k=1$ can be rephrased as
\[
  \p \left( \E \widehat{X}_1 \in \left[ \left( \widehat{X}_1 - c \right) \vee 0, \left( \widehat{X}_1 + c \right) \wedge 1 \right] \right) \geq 1 - 2 \exp \left( - 2 c^2 L \right).
\]
 By algebraic manipulation, $\E \widehat{X}_1 \in \left[ \left( \widehat{X}_1 - c \right) \vee 0, \left( \widehat{X}_1 + c \right) \wedge 1 \right]$ is equivalent to $\log \widetilde{N}_1$ being contained in the interval in~\eqref{eq:k1_conf_int}.
\end{proof}

This bound is useful because we can immediately determine a confidence interval for our estimate. To achieve a confidence level of $1-\alpha$, we can choose $c = c \left( \alpha, L \right)$ to satisfy $2 \exp \left( - 2 c^2 L \right) = \alpha$, i.e., choose 
\begin{equation}\label{eq:c}
 c = \sqrt{\frac{\log \left( 2 / \alpha \right)}{2L}}.
\end{equation}
Then the interval in~\eqref{eq:k1_conf_int} with $c$ given by~\eqref{eq:c} has a confidence level of $1- \alpha$.

\subsection{Conditional bounds}\label{sec:bds_cond} 

Next, we present conditional bounds: given the number of samples that did not coalesce in the time interval $\left[0,\left( k-1 \right) \eps N_0 \right]$, what is the error we make when estimating the population size in the time interval $I_k$? The following result is the same as Proposition~\ref{prop:k_cond_intro} but worded more precisely.

\begin{proposition}\label{prop:k_cond}
 For any $c \geq 0$, the probability conditioned on $L \left( 1 - \widehat{S}_{k-1} \right) = \ell$ (i.e., that $\ell$ samples ``survived'' the first $k-1$ intervals) that the logarithm of the effective constant population size, $\log \widetilde{N}_k$, is in the interval
\begin{equation}\label{eq:k_cond_int}
 \left[ \log \left( \eps N_0 \right) - \log \left( - \log \left( \left( 1 - \frac{L}{\ell}\widehat{X}_k - c \right) \vee 0 \right) \right), \log \left( \eps N_0 \right) - \log \left( - \log \left( \left( 1 - \frac{L}{\ell}\widehat{X}_k + c \right) \wedge 1 \right) \right) \right]
\end{equation}
is at least $1 - 2 \exp \left( - 2 c^2 \ell \right)$.
\end{proposition}
Note that the interval in~\eqref{eq:k_cond_int} is an interval around our estimate $\log \widehat{N}_k$, given  $L \left( 1 - \widehat{S}_{k-1} \right) = \ell$.
\begin{proof}
 Let $\widehat{Y}_k := \frac{L}{\ell} \widehat{X}_k$. 
Given $L \left( 1 - \widehat{S}_{k-1} \right) = \ell$, $\widehat{Y}_k$ is the average of $\ell$ i.i.d.\ indicator variables. 
Therefore Chernoff's bound gives that
\[
 \p \left( \left| \widehat{Y}_k - \E \widehat{Y}_k \right| \geq c \, \middle| \, L \left( 1 - \widehat{S}_{k-1} \right) = \ell \right) \leq 2 \exp \left( - 2 c^2 \ell \right).
\]
In other words,
\[
 \p \left( \E \widehat{Y}_k \in \left[ \left( \widehat{Y}_k - c \right) \vee 0, \left( \widehat{Y}_k + c \right) \wedge 1 \right] \, \middle| \, L \left( 1 - \widehat{S}_{k-1} \right) = \ell \right) \geq 1 - 2 \exp \left( - 2 c^2 \ell \right).
\]
Just as in the proof of Proposition~\ref{prop:k1-simple}, by algebraic manipulation, given  $L \left( 1 - \widehat{S}_{k-1} \right) = \ell$, $\E \widehat{Y}_k \in \left[ \left( \widehat{Y}_k - c \right) \vee 0, \left( \widehat{Y}_k + c \right) \wedge 1 \right]$ is equivalent to $\log \widetilde{N}_k$ being contained in the interval in~\eqref{eq:k_cond_int}.
\end{proof}

Again, this bound is useful because we can immediately determine a confidence interval for our estimate. To achieve a confidence level of $1-\alpha$, we can choose $c = c \left( \alpha, \ell \right)$ to satisfy $2 \exp \left( - 2 c^2 \ell \right) = \alpha$, i.e., choose 
\begin{equation}\label{eq:c_conf_int}
 c = \sqrt{\frac{\log \left( 2 / \alpha \right)}{2\ell}}.
\end{equation}
Then the interval in~\eqref{eq:k_cond_int} with $c$ given by~\eqref{eq:c_conf_int} has a confidence level of $1- \alpha$.

\section{Simulations}\label{sec:sim} 

We illustrate  our estimation procedure on simulated data for the following settings: 
(1) constant size population, 
(2) piecewise constant size population, and 
(3) a population experiencing recent exponential growth; the last setting being germane to recent human population history (see, e.g.,~\cite{tennessen2012evolution} for a study on how recent accelerated population growth, together with weak purifying selection, can lead to an excess of rare functional variants). 
In each case, we simulate $L$ independent coalescence times and apply our estimation procedure described in Section~\ref{sec:algo} with a given $\eps$ to the data; the outcome is summarized in Figures~\ref{fig:constant}--\ref{fig:reconstruct} below. Each figure plots $\log \left( N\left( t \right) / N \left( 0 \right) \right)$ versus $t / N \left( 0 \right)$, i.e., we scale time according to the coalescent timescale, and we plot the population size on a logarithmic scale. 

The true history is the blue line, the estimates over each interval are the red lines, and the confidence intervals at the 95 percent level are given in pink. 
Recall that the logarithm of our estimate~\eqref{eq:est} is
\[
 \log \widehat{N}_k = \log \left( \eps N_0 \right) - \log \left( - \log \left( 1 - \frac{\widehat{X}_k}{1 - \widehat{S}_{k-1}} \right) \right),
\]
and our confidence interval for confidence level $1 - \alpha$ is given by~\eqref{eq:k_cond_int}:
\[
\mbox{
\small
$
\left[ \log \left( \eps N_0 \right) - \log \left( - \log \left( \left( 1 - \frac{\widehat{X}_k}{1 - \widehat{S}_{k-1}} - c \right) \vee 0 \right) \right), \log \left( \eps N_0 \right) - \log \left( - \log \left( \left( 1 - \frac{\widehat{X}_k}{1 - \widehat{S}_{k-1}} + c \right) \wedge 1 \right) \right) \right],
$
}
\]
where
\[
 c = \sqrt{\frac{\log \left( 2 / \alpha \right)}{2L \left( 1 - \widehat{S}_{k-1} \right)}}.
\]
In the case that $\widehat{X}_k=0$ we do not give an estimate but can sometimes still obtain a lower bound on the confidence interval. 
The intervals where the pink extends to the upper or lower margin of the graphing area represent a confidence bound that is infinite or zero, i.e., where the minimum or maximum are taken to be one or zero in the expressions for the confidence bounds above. 

The error $E_{k,1}$ is not represented in the plots, but would add $\delta$ to each side of each confidence interval. Alternatively, we can view our statistic as an estimate of the effective constant population size over the given interval.

The plots in Figure~\ref{fig:reconstruct} compare our lower bounds to the confidence intervals of our estimation procedure.
The black lines are the 95 percent uncertainty intervals: 
in a given time interval, 
populations within the interval given by the black lines cannot be distinguished from the true blue line population with the amount of data in hand with probability $.95$. 
Thus if the red line is within the interval given by the black lines, then our estimate is in some sense the best that can be achieved.
Note that when an interval has no upper black line, then there is not enough data to distinguish between a history of the blue line size
in the interval and \emph{any} larger size population with 95\% certainty.

Even though our assumed data is idealized and unrealistic (exact coalescent data at thousands of independent sites),
these simulations allow us to make some general qualitative observations. 
The major determining factor of the performance of our procedure in a time interval is the number of coalescence times that have survived to that interval.
So having more data (i.e., larger $L$) leads to more accurate estimates, and the estimates lose accuracy moving back in time as the number of 
data points decreases. 
Moreover, there is a rare event effect when there are few coalescences in an interval---having no coalescences in an interval is not very informative---and this leads to 
the consistent underestimates in the deepest part of the histories. 
Possibly this effect would be lessened by lengthening the widths of the intervals as they go back in time, although this would smooth out
big features in the history.
Also note that in the presence of a  bottleneck, i.e., a time period where the population 
becomes small, there are many coalescences due to the increased rate. In turn this decreases the number of available data deeper in history. For example, compare the accuracy of the estimates
at time $t/N(0)=4$ in the constant population of Figure~\ref{fig:constant} to that of the piecewise constant population of Figure~\ref{fig:piece} where there is a bottleneck starting around time $t/N(0)=1$. 
Finally, note that in Figure~\ref{fig:reconstruct} the confidence intervals of our procedure and the lower bound black lines are rather tight in the presence of a significant amount of data, but loosen as the number of data points decreases.
 \section{Discussion}\label{sec:discussion} 

An assortment of methods have been developed to infer a population's history
from (an ever increasing amount of) genetic data \cite{Bhaskar2014, Drummond2005, Excoffier2013, harris2013inferring, heled2008bayesian,   Li2011, Nielsen2000, palamara2012length, Sheehan2013}. 
These methods are necessarily computational and approximate
and so the quality of the outputs of these methods cannot be rigorously justified.    
However, understanding the theoretical limitations of inferring past population history \cite{bhaskar2013identifiability, myers2008can} 
is of the utmost importance, 
since it is becoming increasingly common for such analyses to be used as the main tool for inference, with less
emphasis on external verification (e.g., fossil or paleontological record)~\cite{Bos2014}.

Here we have provided lower bounds on the amount of idealized data needed to infer a population history to a given accuracy. 
Our bounds should be considered as underestimates of the amount of data necessary for inference in methods which use sequence data, 
so they can be used as a guide when performing such analyses. 
We end with some further avenues of study and open problems.

\subsection{Open Problems}\label{sec:open} 

\noindent \textbf{$n$-coalescence trees.} 
With the exception of Theorem~\ref{lem:bdm_intro}, we assume that our data are $L$ i.i.d.\ coalescence times between pairs of individuals. 
If instead our data are $L$ i.i.d.\ coalescence trees among $n$ individuals then 
how does this affect the bounds? In the setting of Theorem~\ref{lem:bdm_intro} when comparing two constant populations, increasing the number of individuals $n$ in
the coalescent tree is as good as increasing the number of independent loci $L$. This shouldn't hold true in general since adding individuals does not greatly increase the depth of the tree.
For moderate values of $n$, we expect that estimates of the deep history will not be greatly affected since the time of coalescence for the final two lineages is roughly half the length of the coalescent tree started from \emph{infinitely} many individuals. On the other hand, explosive growth in the near history should be estimated better using more individuals since the amount of coalescing in the near past will increase.
It would be interesting to better understand how increasing the number of individuals in the tree affects the lower bounds of Section~\ref{sec:low} and the upper bounds 
provided by some generalization of our inference algorithm of Section~\ref{sec:results}. For some discussion on the affect of increasing the size of the tree versus increasing 
independent loci, see~\cite{heled2008bayesian}.

\smallskip

\noindent \textbf{Estimation from sequence data.} 
The assumption that we know exact coalescence times is unrealistic.
These times  need to be estimated from sequence data at independent loci with good accuracy. 
How do we estimate coalescence times from sequence data with quantitative upper and lower bounds analogous to those here? 
%
%
%
\smallskip

\noindent \textbf{Population substructure.}  Assume we want to estimate
a population that is not only changing over time, but also has sub-populations that merge
and split, and which may have migration rates between them. Are there analogs of our results
in this setting? Note that identifiability can be an issue here since, for example,
a constant population that splits at a given point in the past has the same 
distribution of coalescence times among two individuals as a single population that grows exponentially at a specific rate
(backward in time) starting at the time of the split.

%
%
%
%
%
%
%
%
%


\section*{Acknowledgments}

We thank Anand Bhaskar, Luke Gandolfo, Jasmine Nirody, Sara Sheehan, and Yun Song for helpful discussions and relevant references. 
A portion of the work for this project was completed when NR was at University of California, Berkeley with support from NSF grants DMS-0704159, DMS-0806118, DMS-1106999 and ONR grant N00014-11-1-0140.


\bibliographystyle{abbrv}
\bibliography{Density}

\begin{figure}
    \centering
    \begin{subfigure}[h]{0.4\textwidth}
	\centering
	\includegraphics[width=\textwidth]{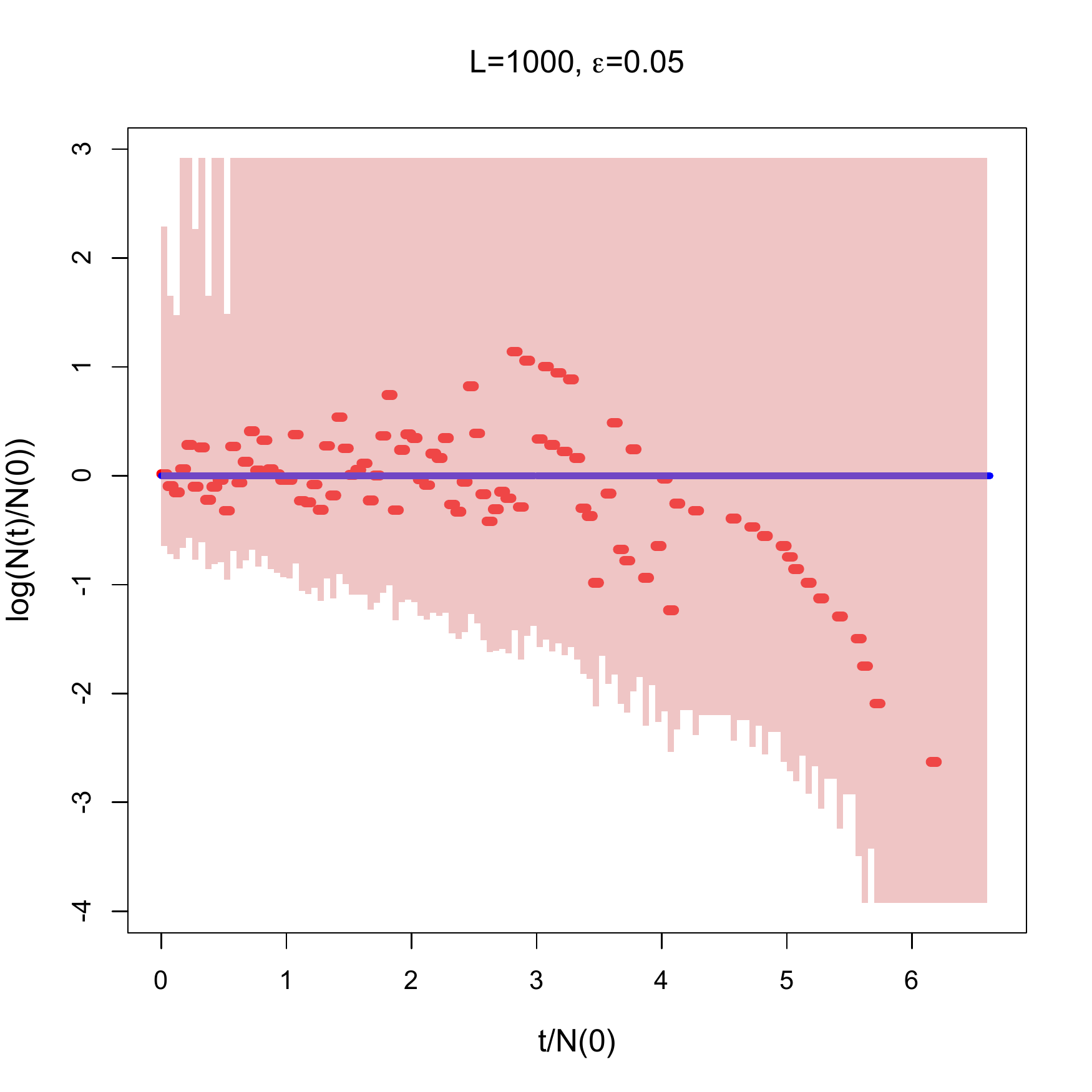}
	\caption{$L=10^3, \eps = 0.05$.}
	\label{fig:const1}
    \end{subfigure}
    \ \ 
    \begin{subfigure}[h]{0.4\textwidth}
	\centering
	\includegraphics[width=\textwidth]{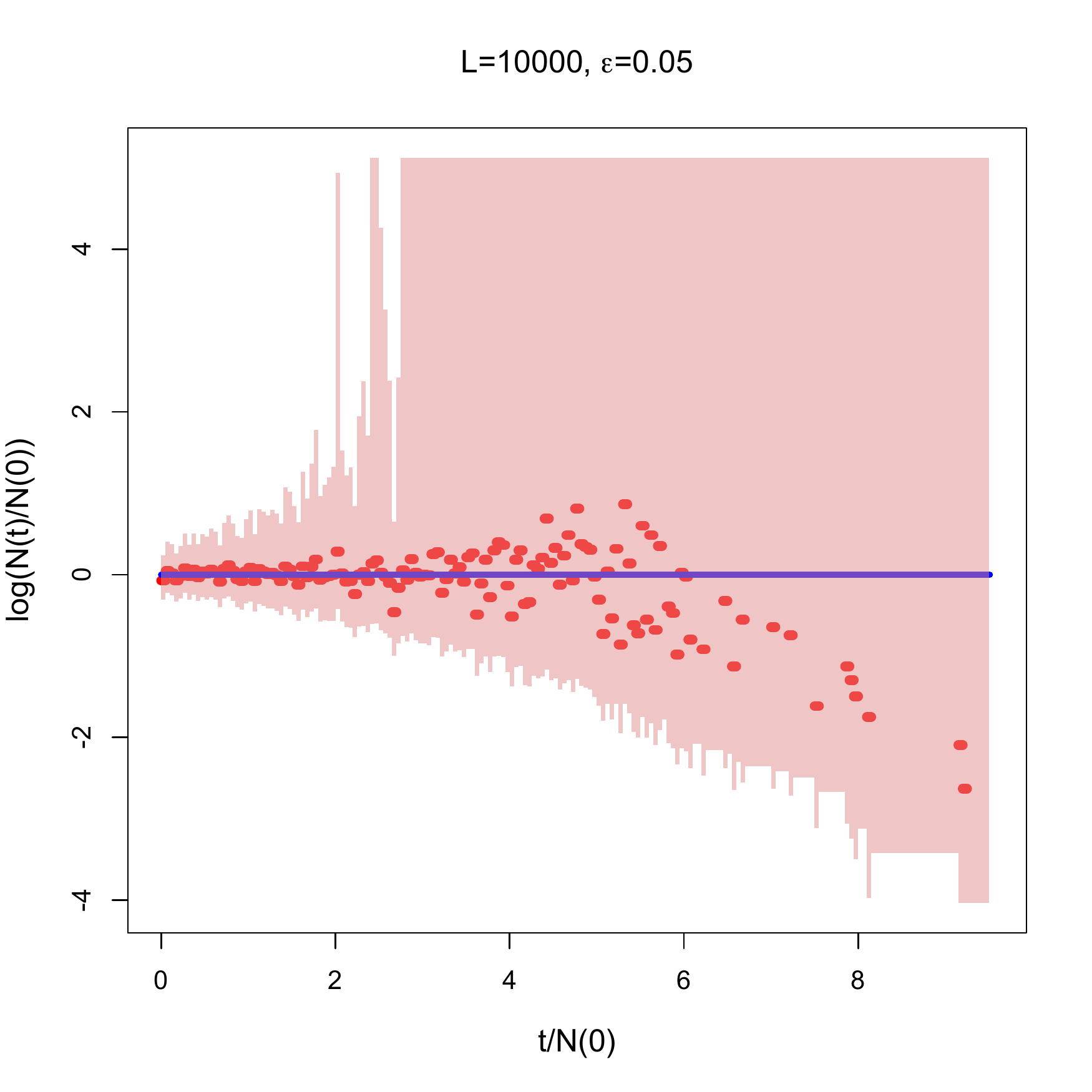}
	\caption{$L=10^4, \eps = 0.05$.}
	\label{fig:const2}
    \end{subfigure}
    \\ 
      \begin{subfigure}[h]{0.4\textwidth}
	\centering
	\includegraphics[width=\textwidth]{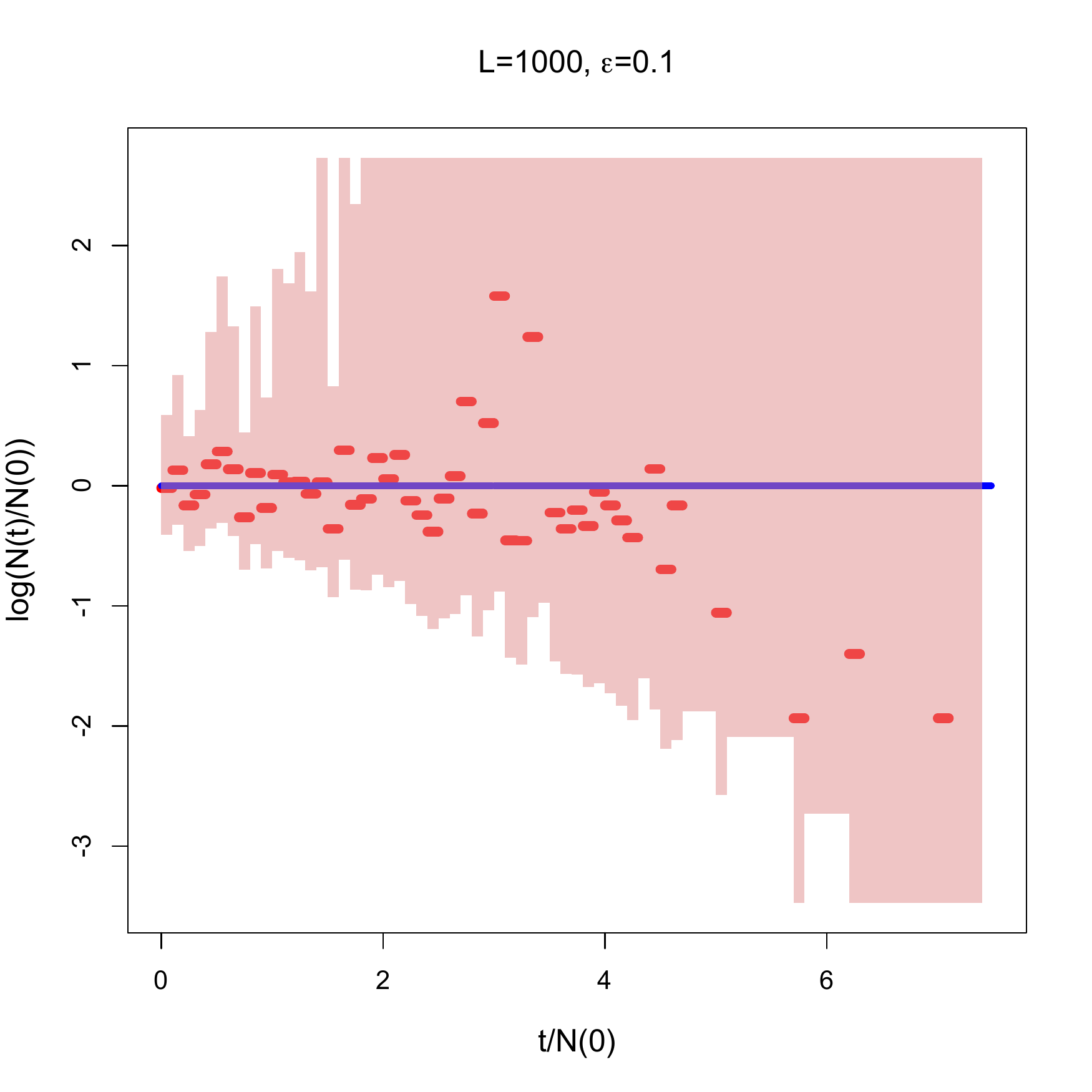}
	\caption{$L=10^3, \eps = 0.1$.}
	\label{fig:const3}
    \end{subfigure}
    \ \ 
    \begin{subfigure}[h]{0.4\textwidth}
	\centering
	\includegraphics[width=\textwidth]{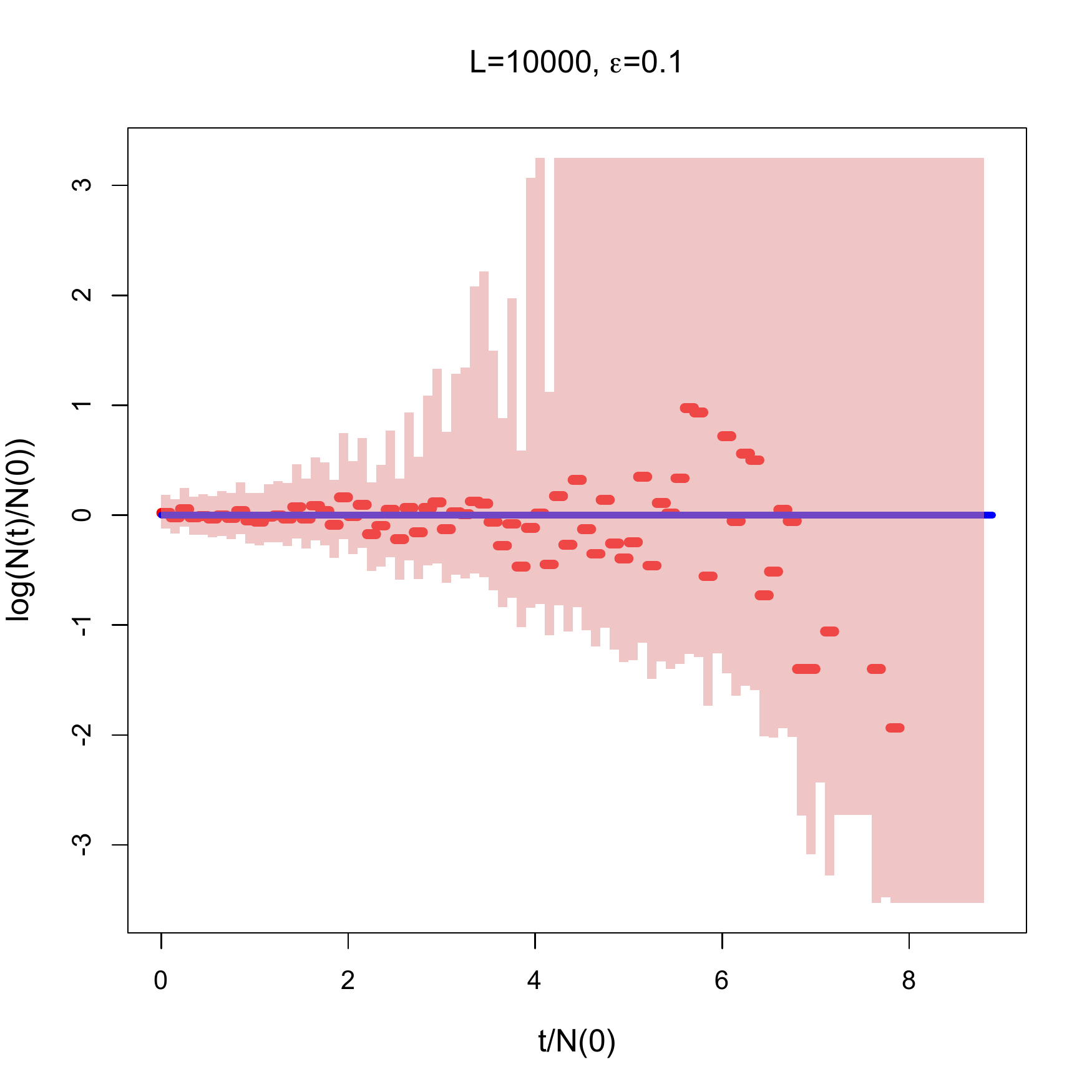}
	\caption{$L=10^4, \eps = 0.1$.}
	\label{fig:const4}
    \end{subfigure}
    \\
      \begin{subfigure}[h]{0.4\textwidth}
	\centering
	\includegraphics[width=\textwidth]{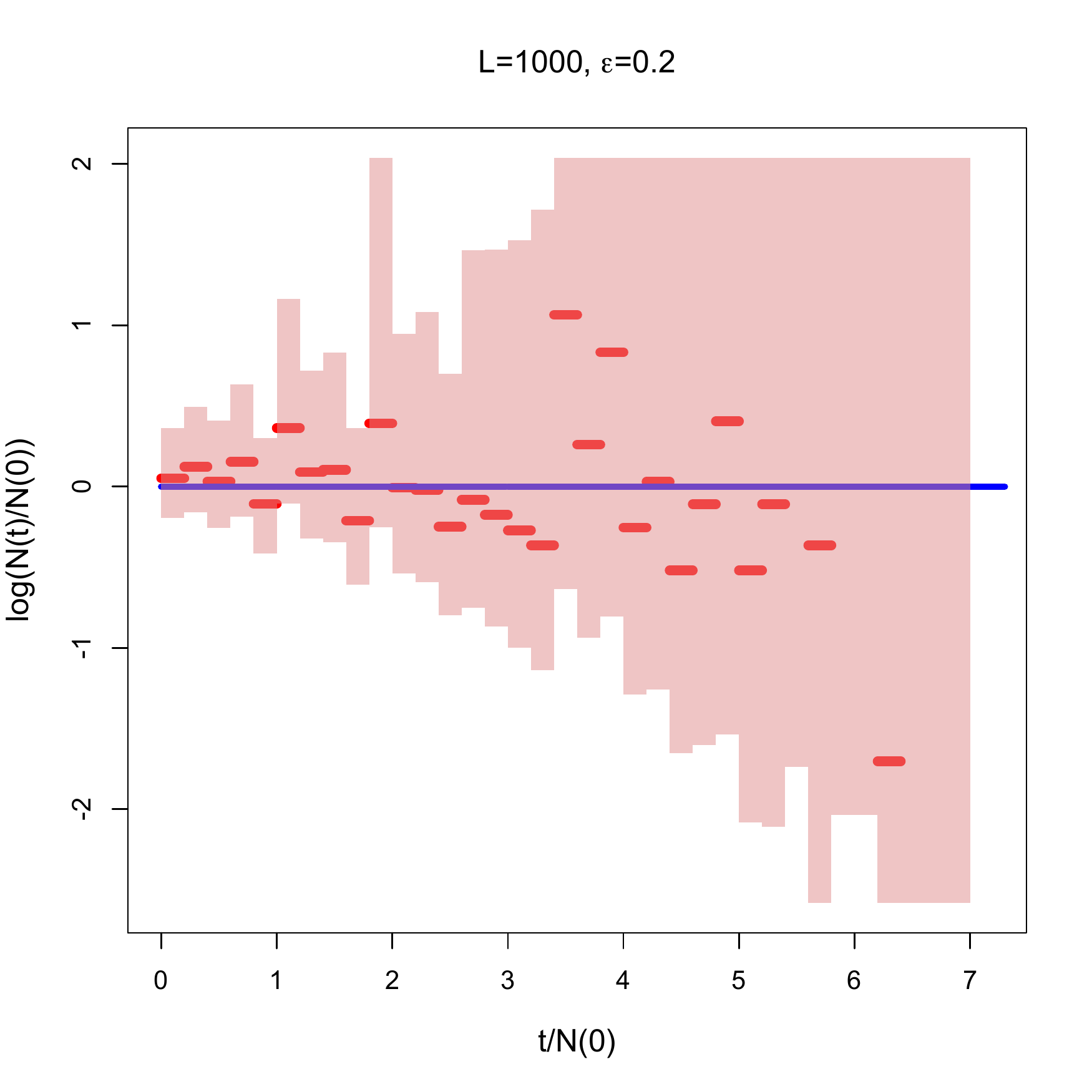}
	\caption{$L=10^3, \eps = 0.2$.}
	\label{fig:const5}
    \end{subfigure}
    \ \ 
    \begin{subfigure}[h]{0.4\textwidth}
	\centering
	\includegraphics[width=\textwidth]{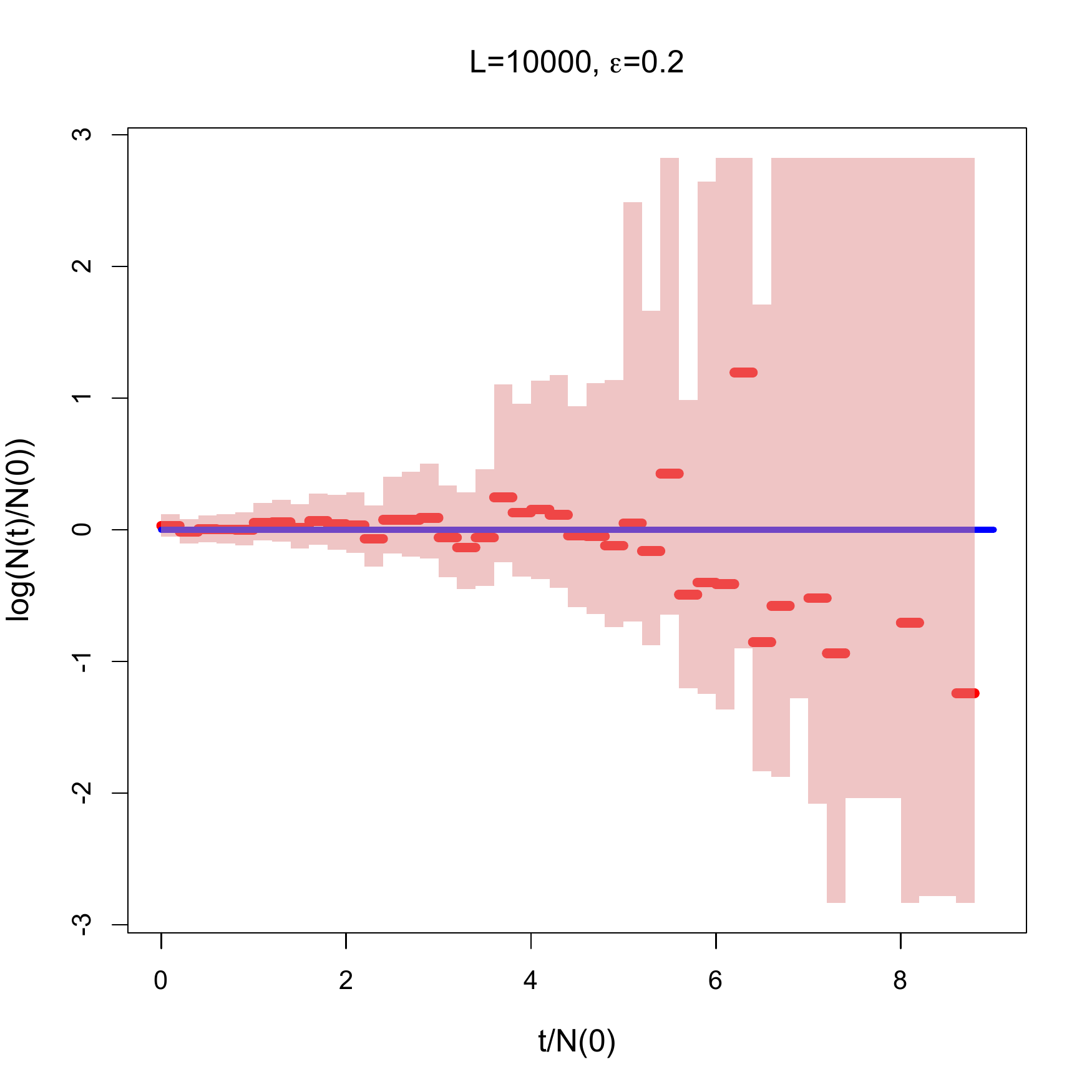}
	\caption{$L=10^4, \eps = 0.2$.}
	\label{fig:const6}
    \end{subfigure}

    \caption{Estimating a constant population size.}
    \label{fig:constant}
\end{figure}

\begin{figure}
    \centering
    \begin{subfigure}[h]{0.4\textwidth}
	\centering
	\includegraphics[width=\textwidth]{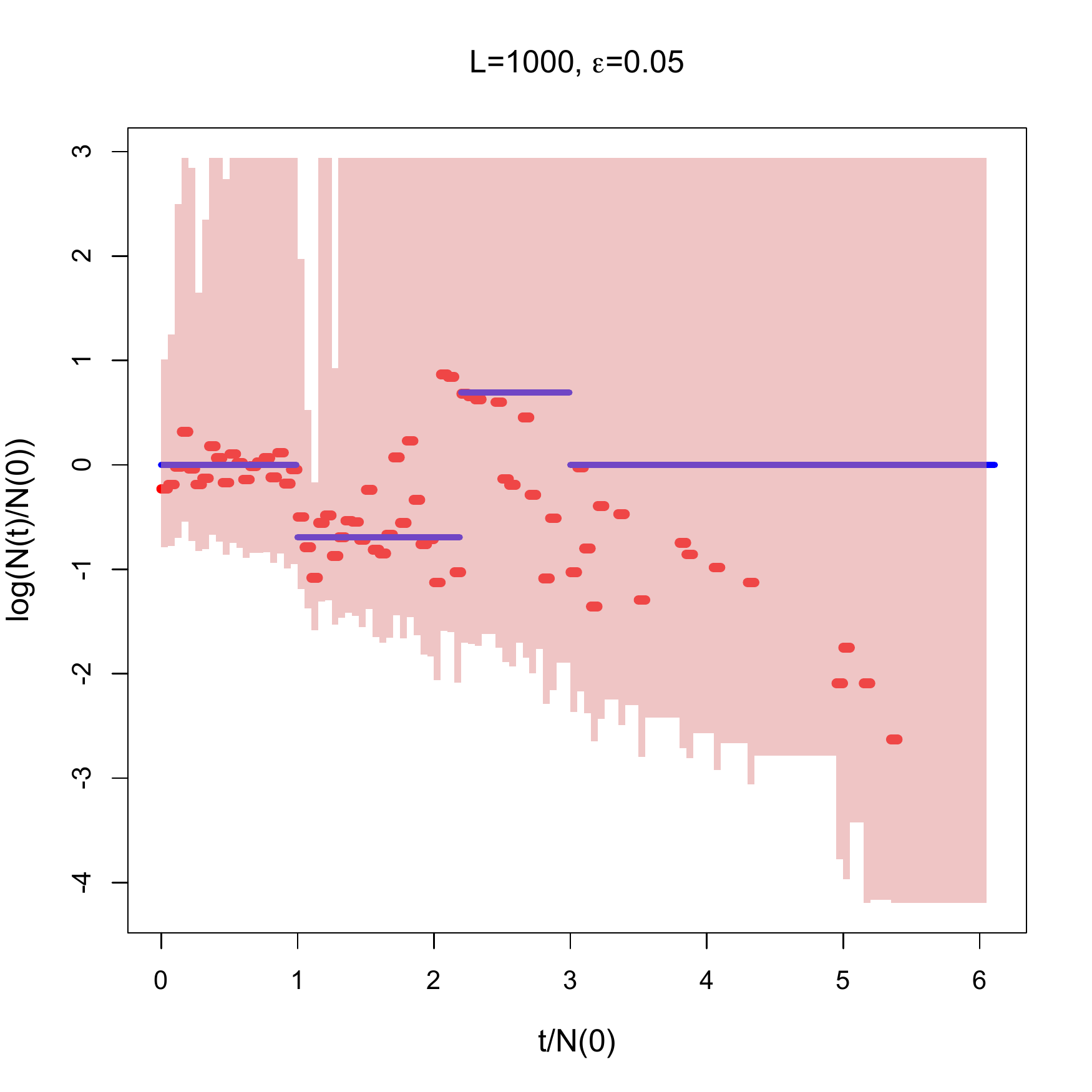}
	\caption{$L=10^3, \eps = 0.05$.}
	\label{fig:piece1}
    \end{subfigure}
    \ \ 
    \begin{subfigure}[h]{0.4\textwidth}
	\centering
	\includegraphics[width=\textwidth]{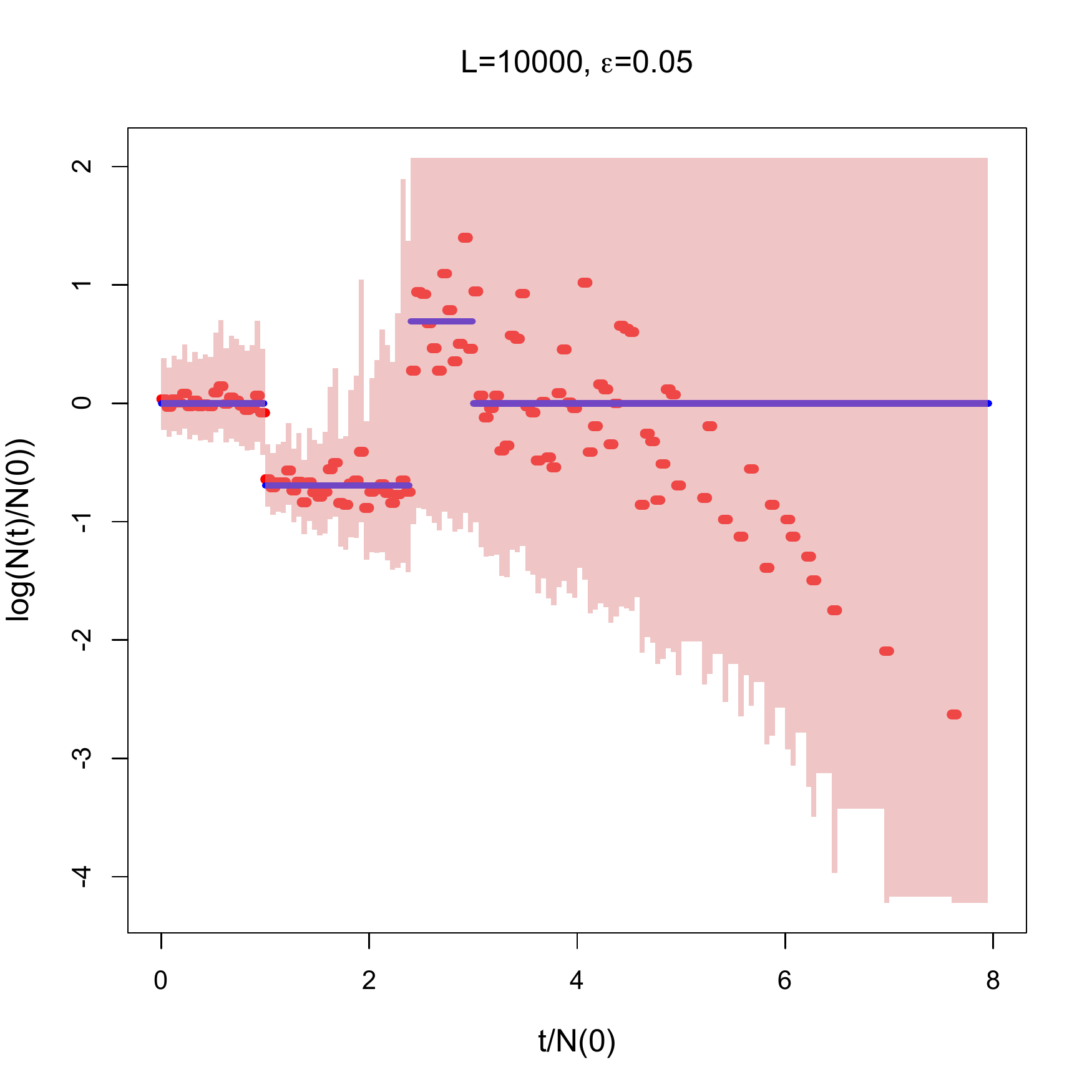}
	\caption{$L=10^4, \eps = 0.05$.}
	\label{fig:piece2}
    \end{subfigure}
    \\ 
      \begin{subfigure}[h]{0.4\textwidth}
	\centering
	\includegraphics[width=\textwidth]{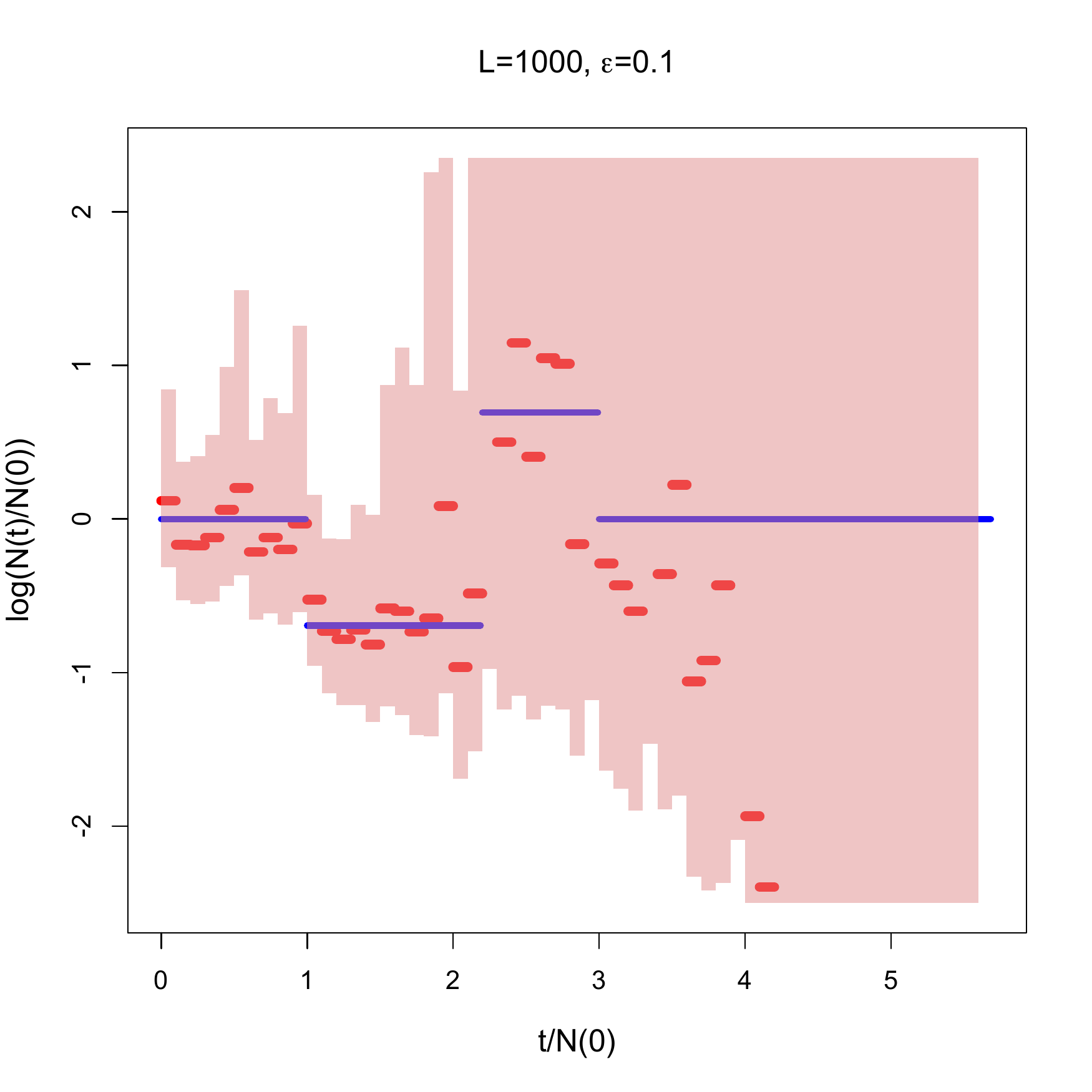}
	\caption{$L=10^3, \eps = 0.1$.}
	\label{fig:piece3}
    \end{subfigure}
    \ \ 
    \begin{subfigure}[h]{0.4\textwidth}
	\centering
	\includegraphics[width=\textwidth]{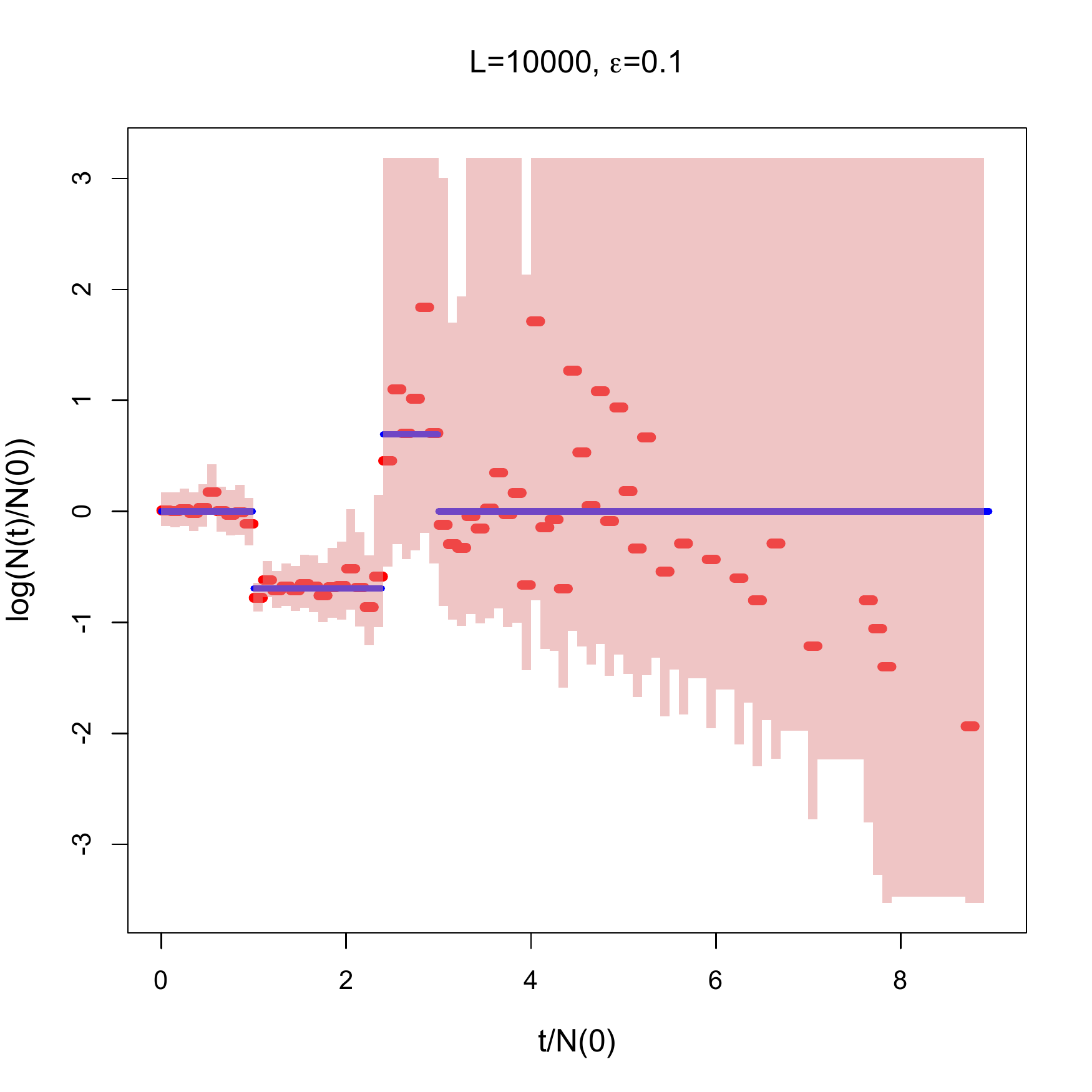}
	\caption{$L=10^4, \eps = 0.1$.}
	\label{fig:piece4}
    \end{subfigure}
    \\
      \begin{subfigure}[h]{0.4\textwidth}
	\centering
	\includegraphics[width=\textwidth]{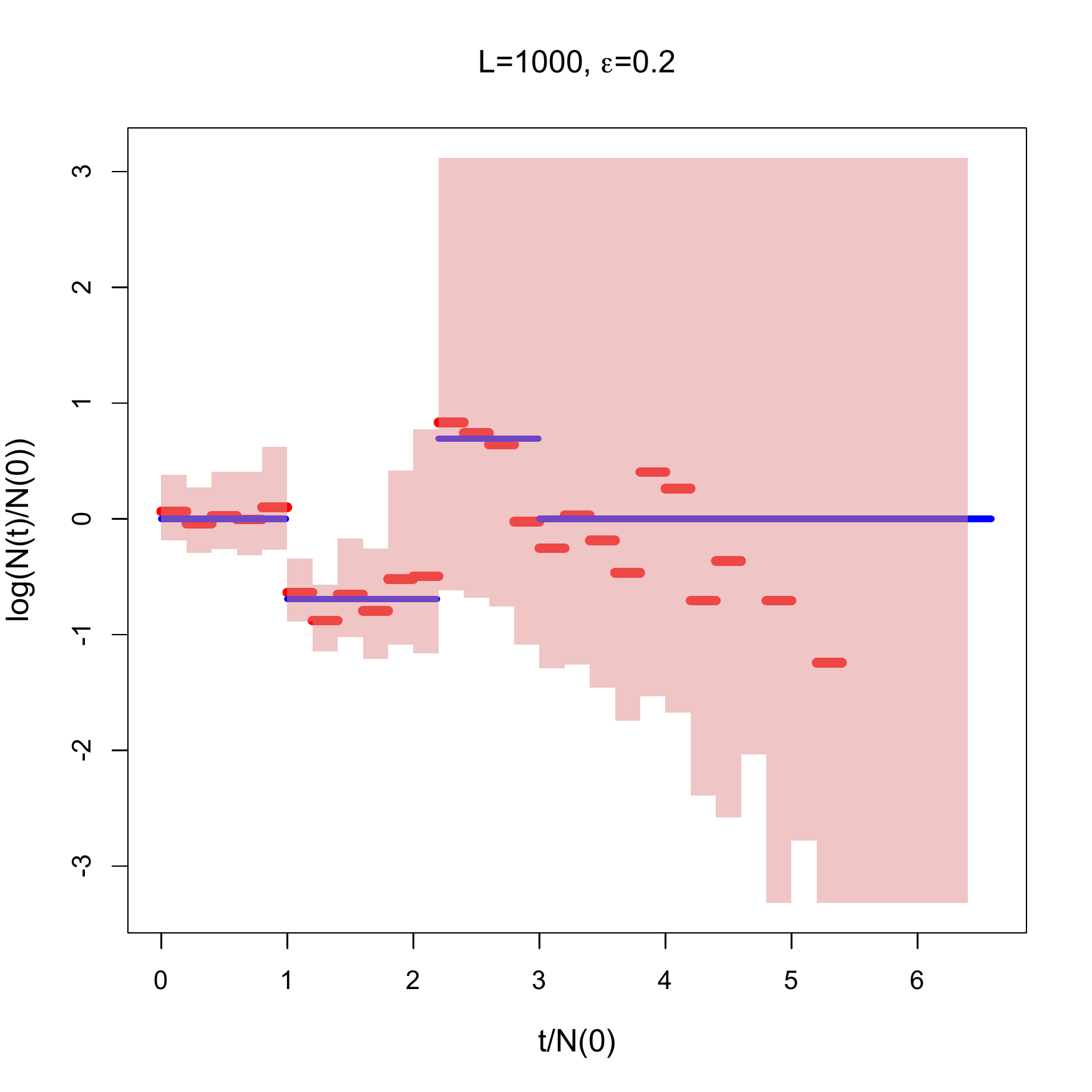}
	\caption{$L=10^3, \eps = 0.2$.}
	\label{fig:piece5}
    \end{subfigure}
    \ \ 
    \begin{subfigure}[h]{0.4\textwidth}
	\centering
	\includegraphics[width=\textwidth]{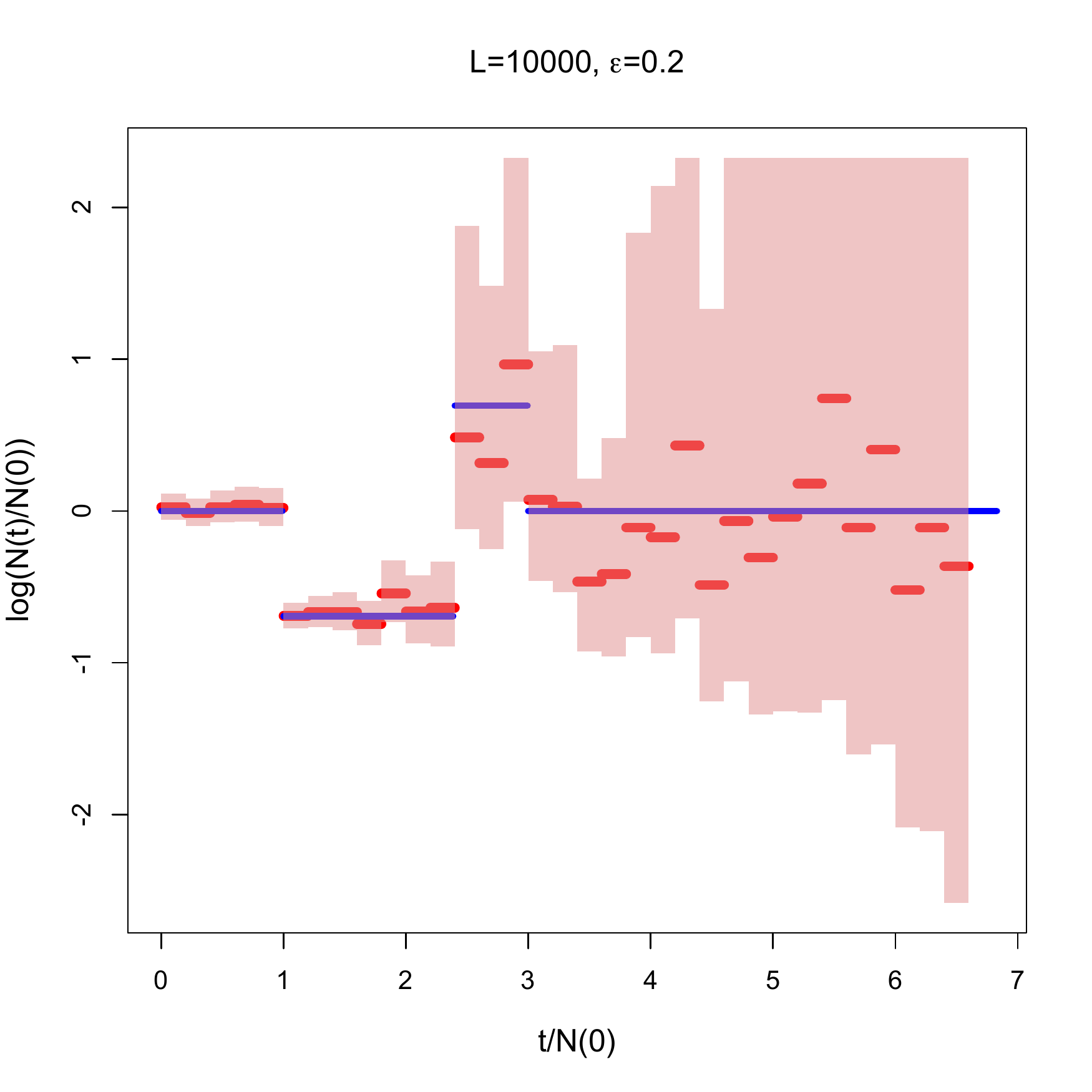}
	\caption{$L=10^4, \eps = 0.2$.}
	\label{fig:piece6}
    \end{subfigure}

    \caption{Estimating a piecewise constant population size.}
    \label{fig:piece}
\end{figure}

\begin{figure}
    \centering
    \begin{subfigure}[h]{0.4\textwidth}
	\centering
	\includegraphics[width=\textwidth]{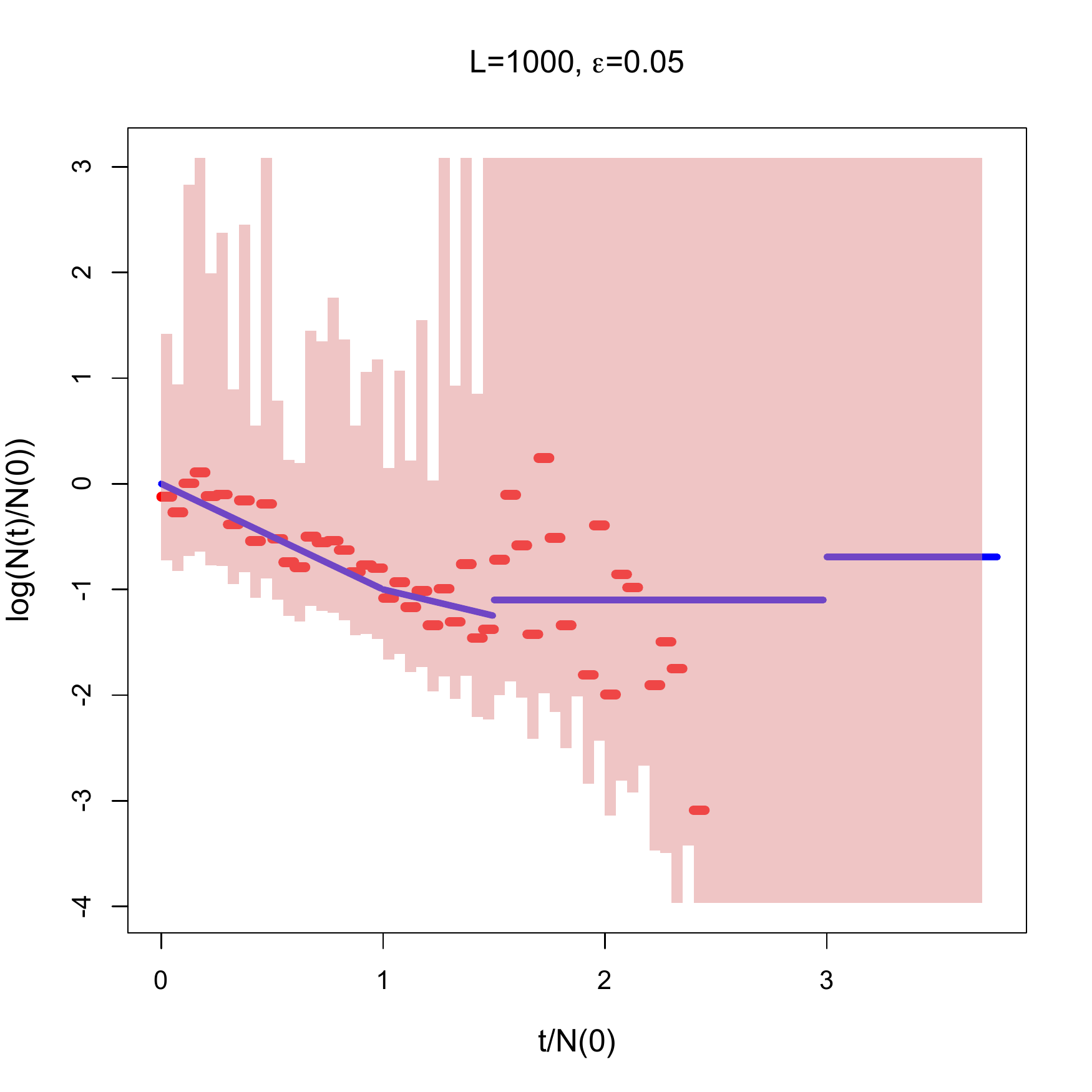}
	\caption{$L=10^3, \eps = 0.05$.}
	\label{fig:piex1}
    \end{subfigure}
    \ \ 
    \begin{subfigure}[h]{0.4\textwidth}
	\centering
	\includegraphics[width=\textwidth]{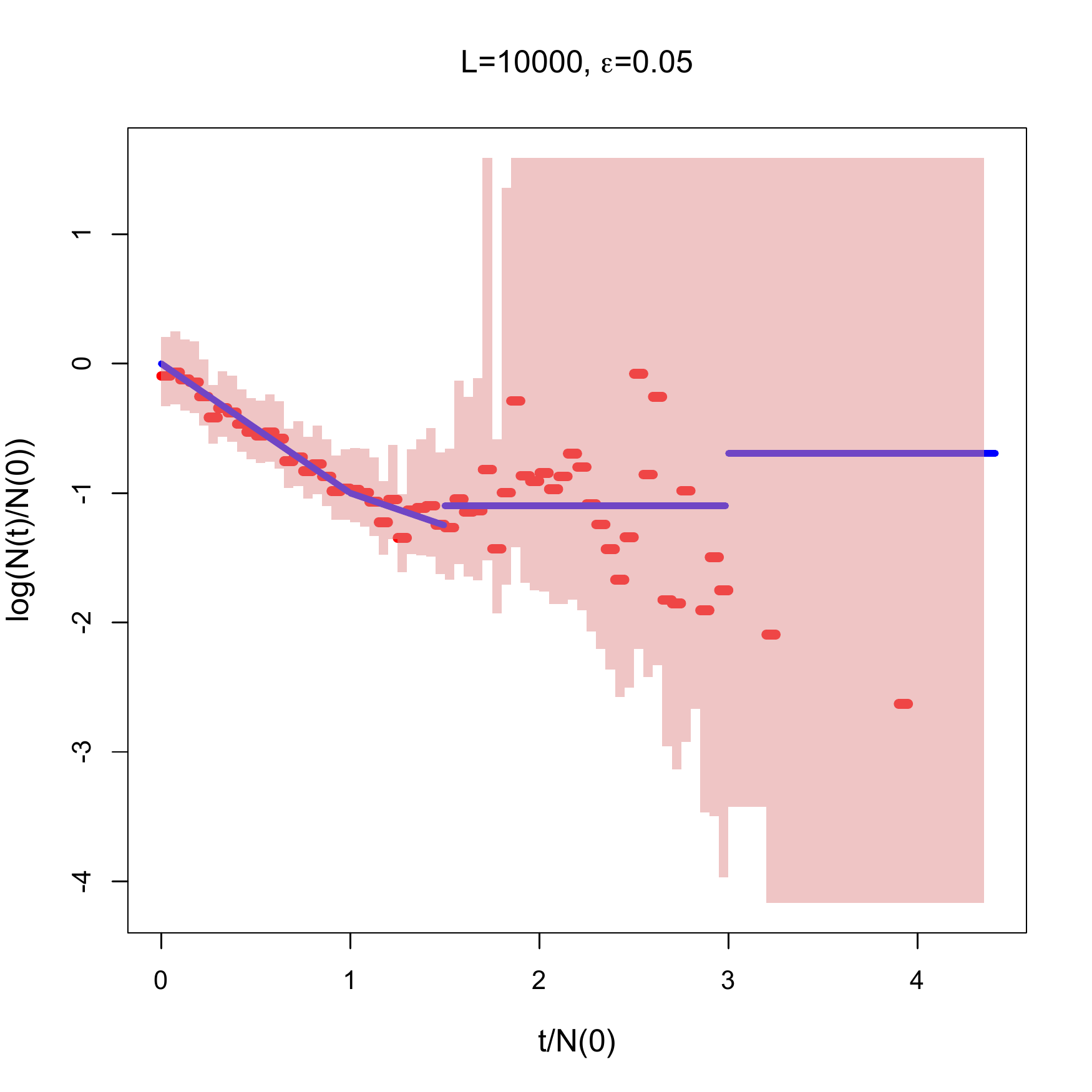}
	\caption{$L=10^4, \eps = 0.05$.}
	\label{fig:piex2}
    \end{subfigure}
    \\ 
      \begin{subfigure}[h]{0.4\textwidth}
	\centering
	\includegraphics[width=\textwidth]{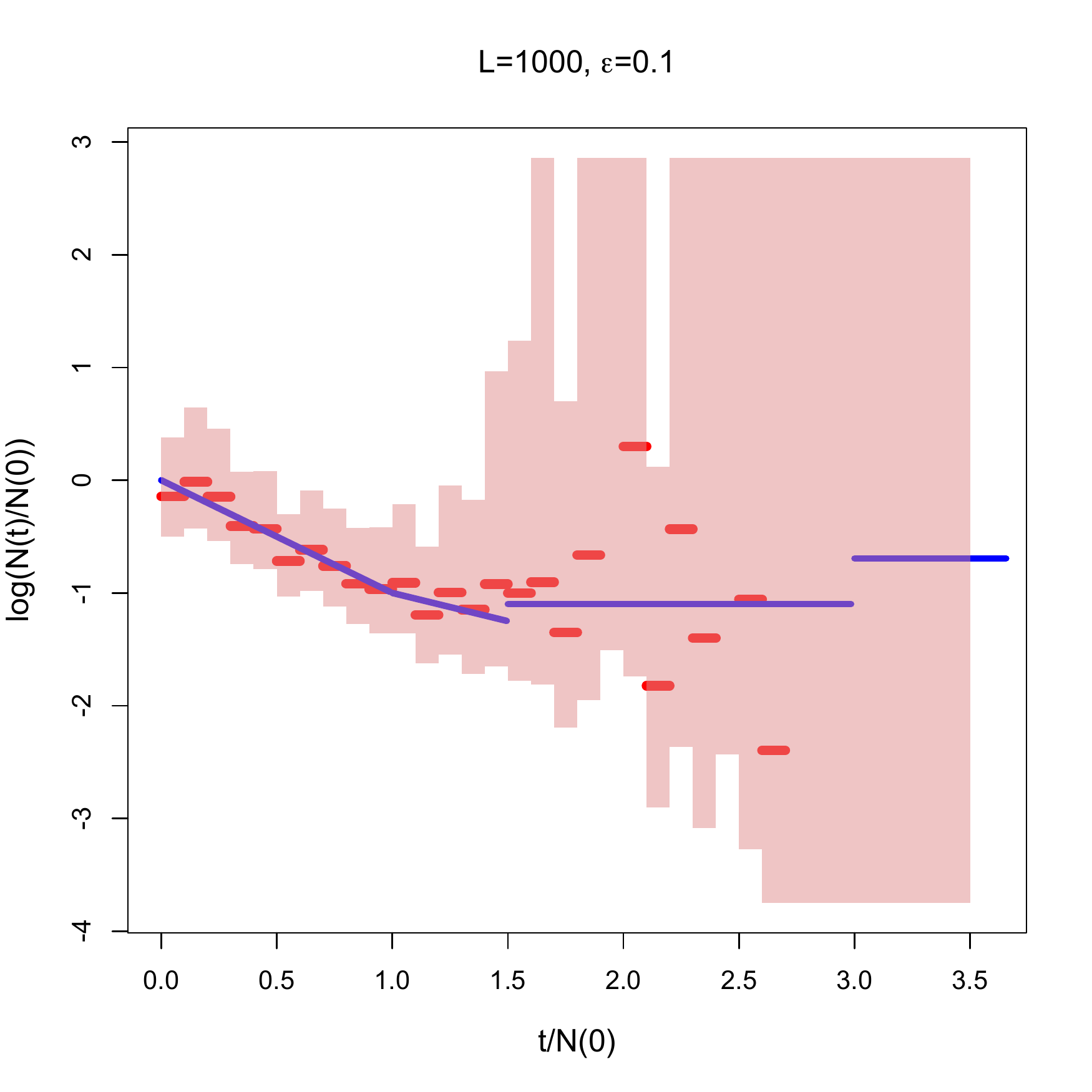}
	\caption{$L=10^3, \eps = 0.1$.}
	\label{fig:piex3}
    \end{subfigure}
    \ \ 
    \begin{subfigure}[h]{0.4\textwidth}
	\centering
	\includegraphics[width=\textwidth]{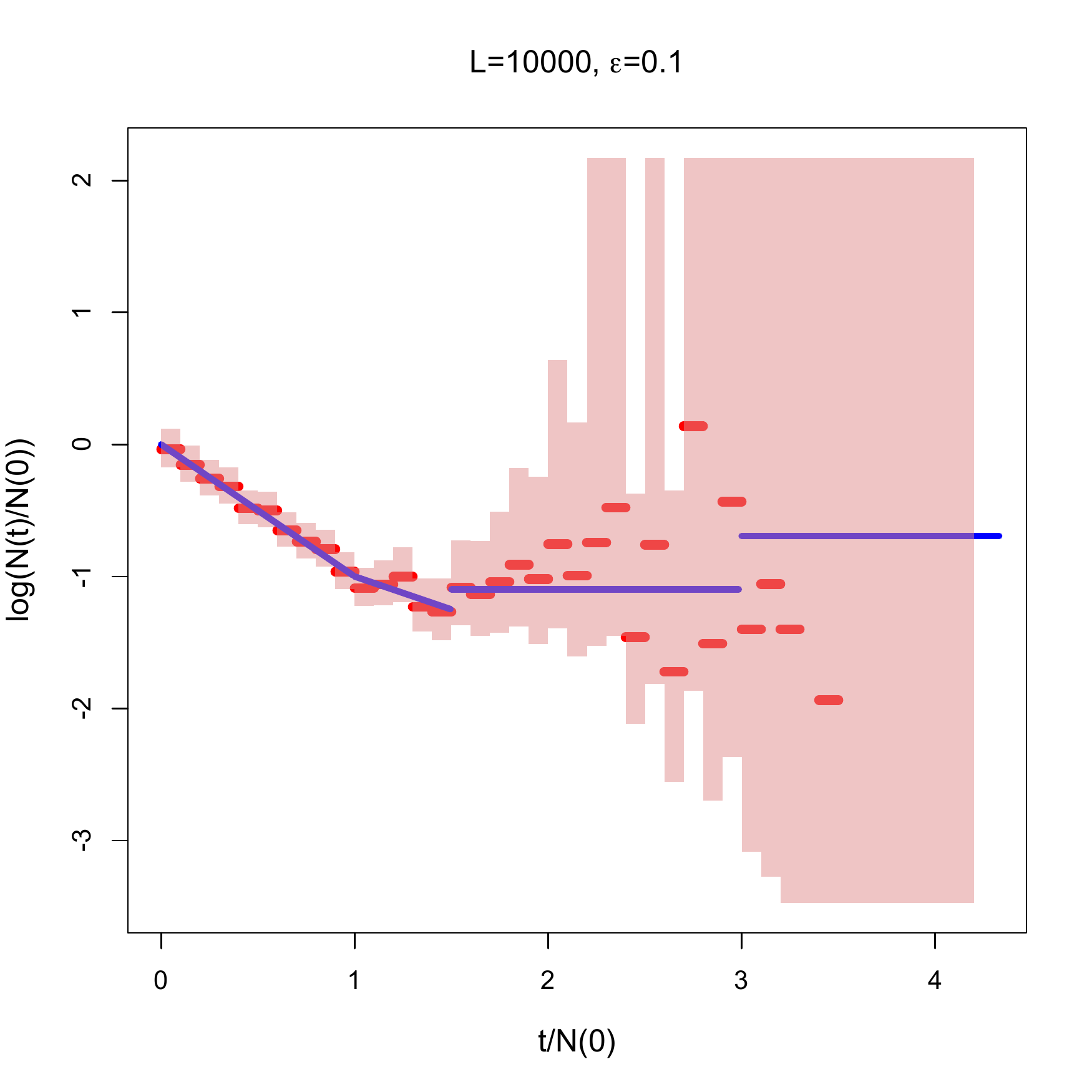}
	\caption{$L=10^4, \eps = 0.1$.}
	\label{fig:piex4}
    \end{subfigure}
    \\
      \begin{subfigure}[h]{0.4\textwidth}
	\centering
	\includegraphics[width=\textwidth]{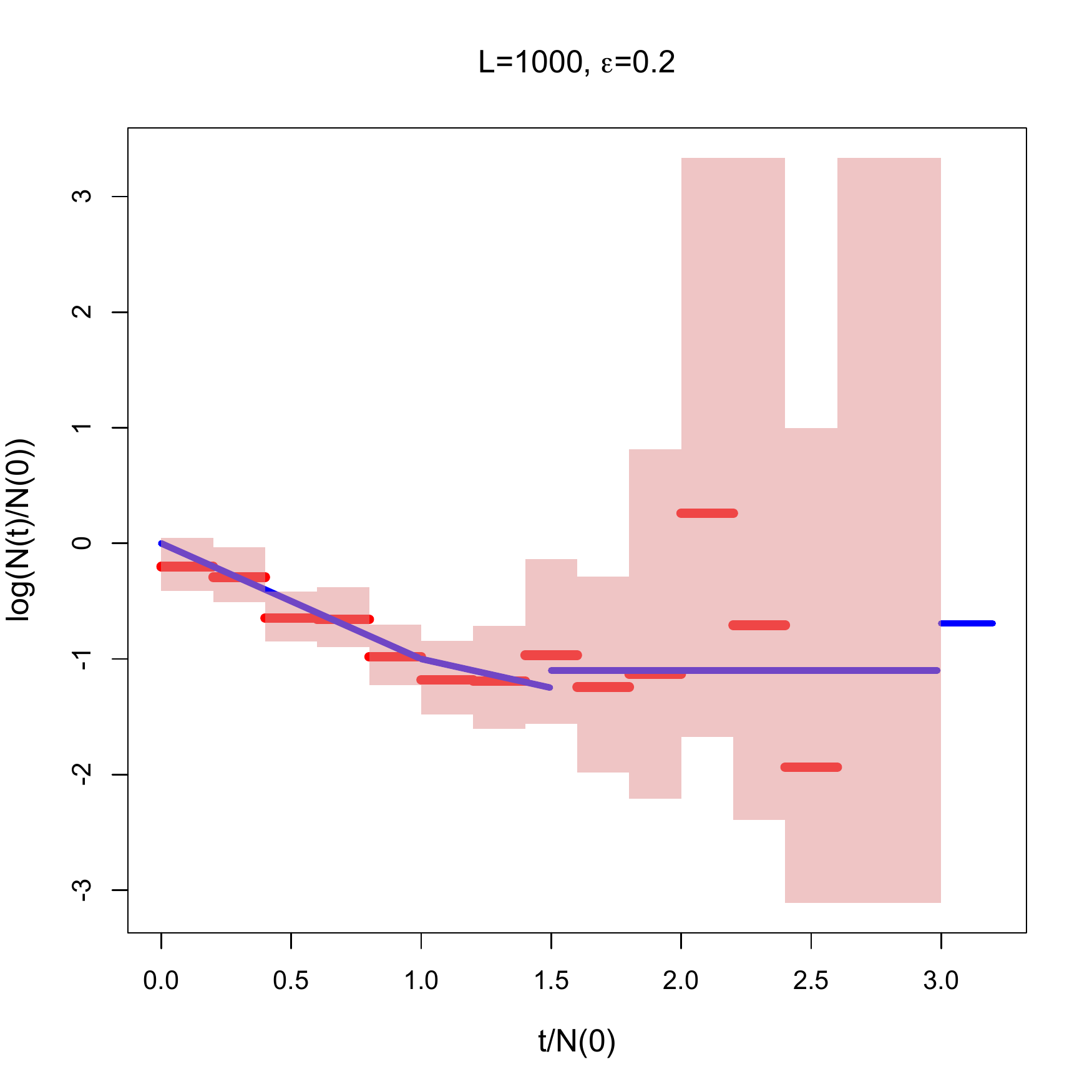}
	\caption{$L=10^3, \eps = 0.2$.}
	\label{fig:piex5}
    \end{subfigure}
    \ \ 
    \begin{subfigure}[h]{0.4\textwidth}
	\centering
	\includegraphics[width=\textwidth]{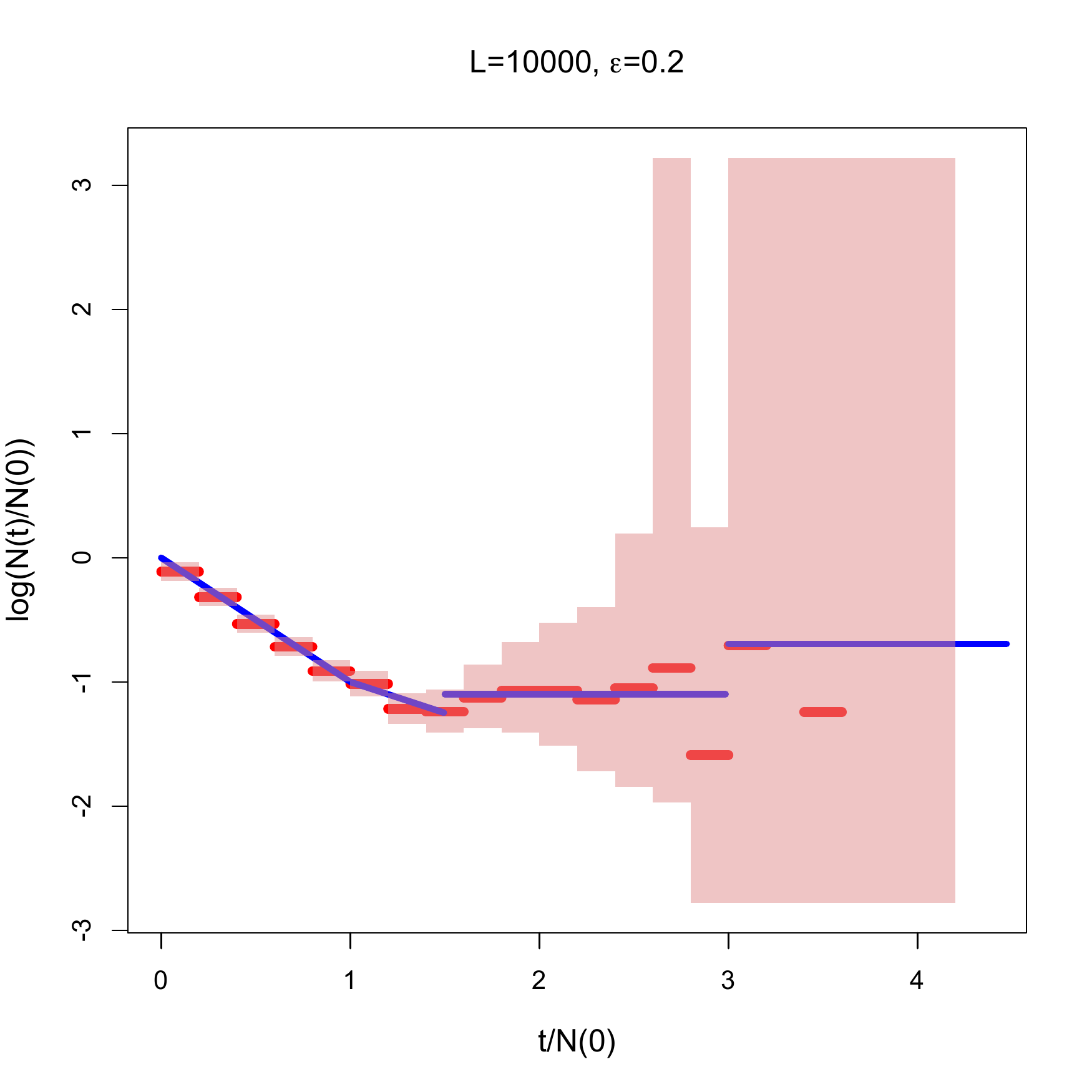}
	\caption{$L=10^4, \eps = 0.2$.}
	\label{fig:piex6}
    \end{subfigure}

    \caption{Estimating a population history with piecewise exponential change. }
    \label{fig:piex}
\end{figure}

\begin{figure}
    \centering
    \begin{subfigure}[h]{0.45\textwidth}
	\centering
	\includegraphics[width=\textwidth] {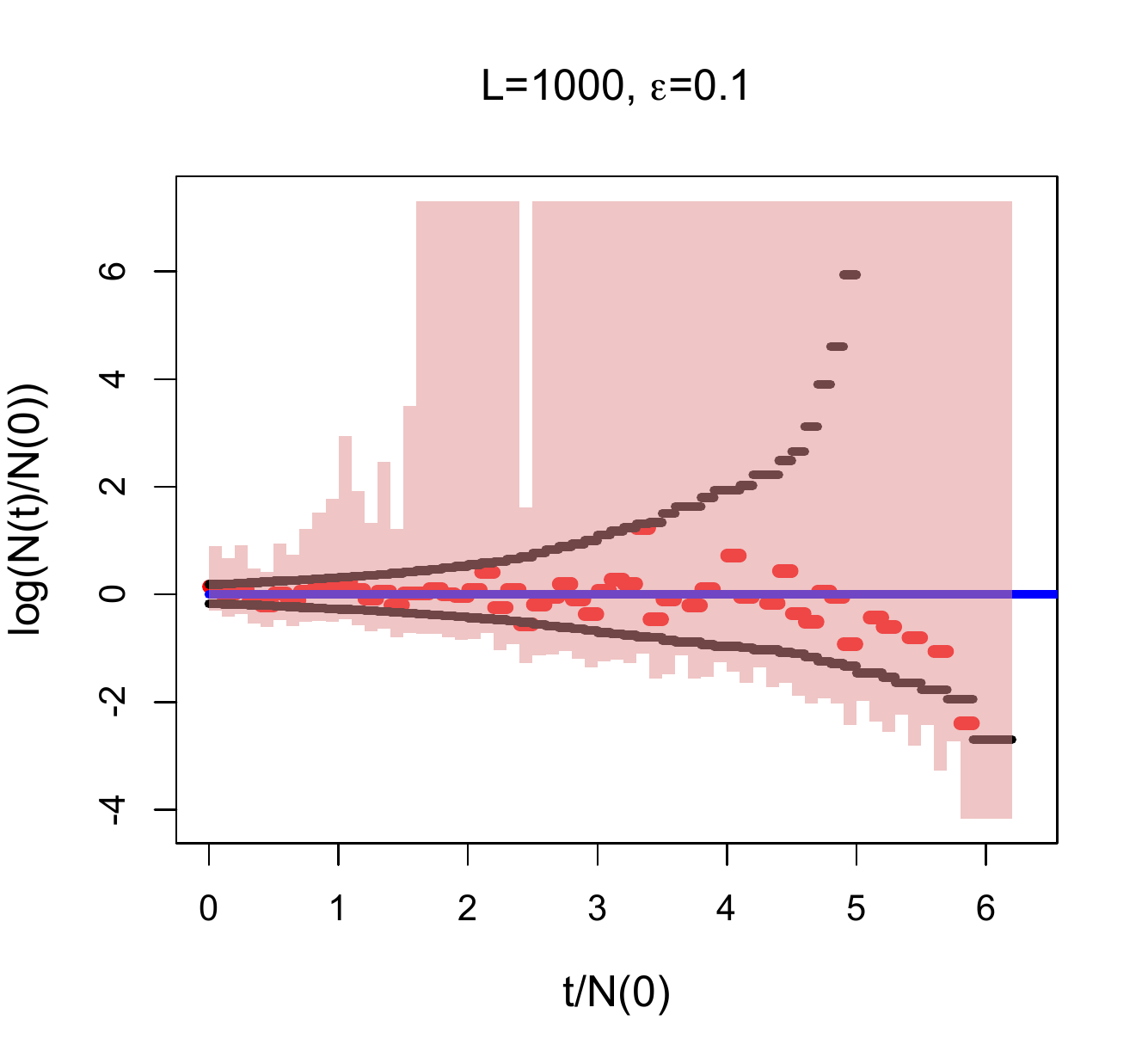}   
	\caption{$L=10^3, \eps = 0.1$.}
	\label{fig:constb1}
    \end{subfigure}
    \ \ 
    \begin{subfigure}[h]{0.45\textwidth}
	\centering
	\includegraphics[width=\textwidth]{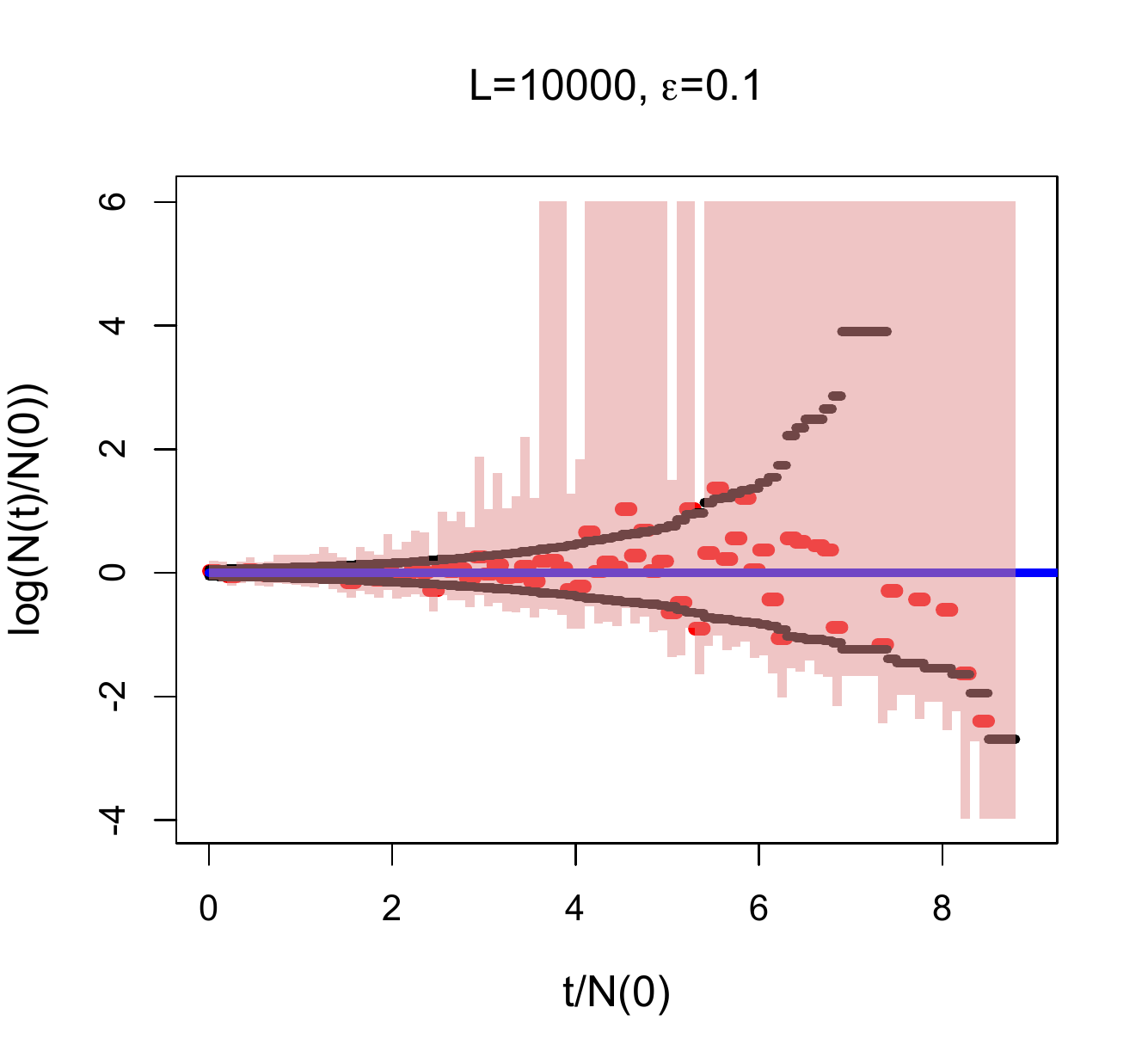}   
	\caption{$L=10^4, \eps = 0.1$.}
	\label{fig:constlb2}
    \end{subfigure}
    \\ 
      \begin{subfigure}[h]{0.45\textwidth}
	\centering
	\includegraphics[width=\textwidth]{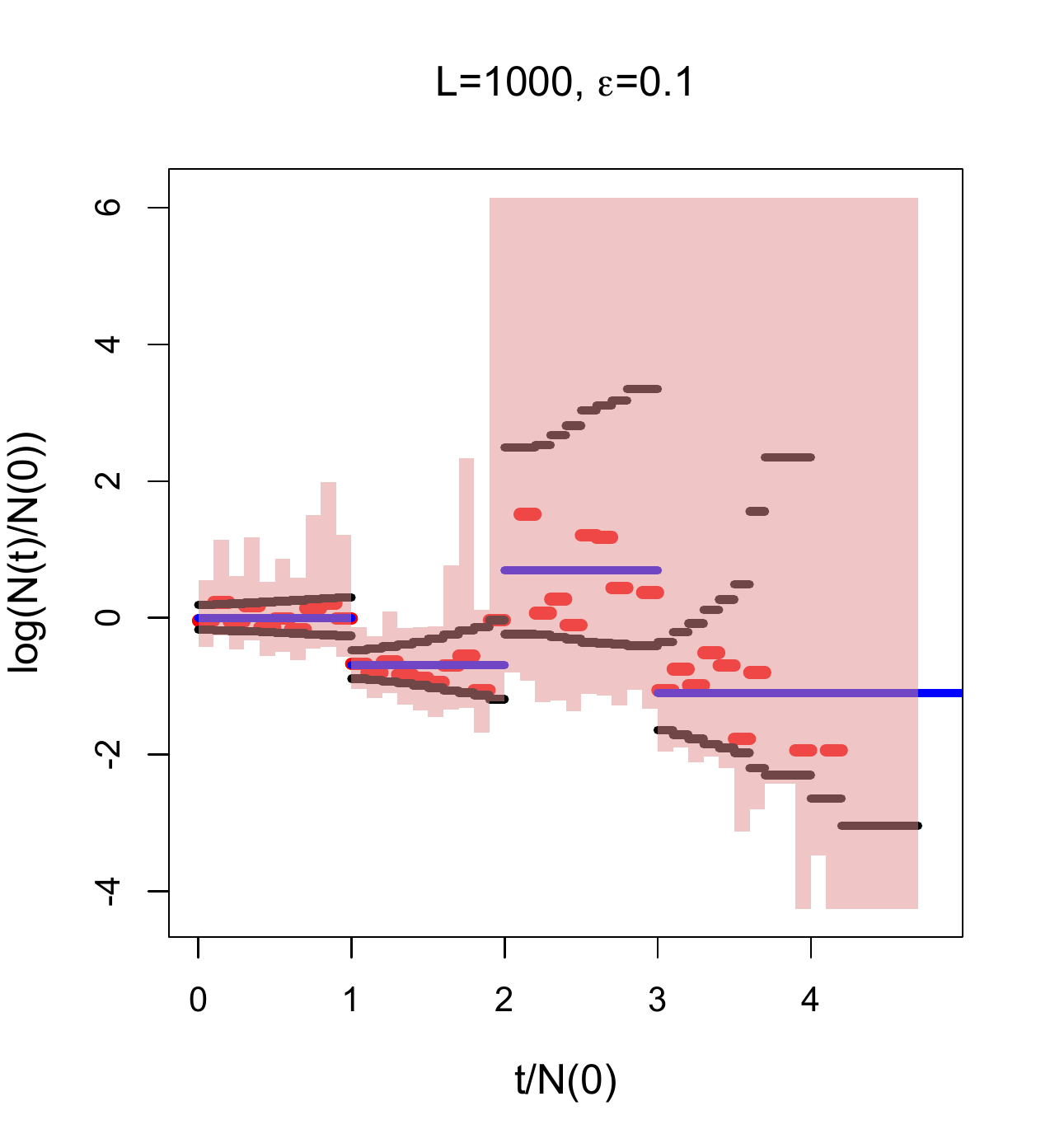}  
	\caption{$L=10^3, \eps = 0.1$.}
	\label{fig:constlb3}
    \end{subfigure}
    \ \ 
    \begin{subfigure}[h]{0.45\textwidth}
	\centering
	\includegraphics[width=\textwidth]{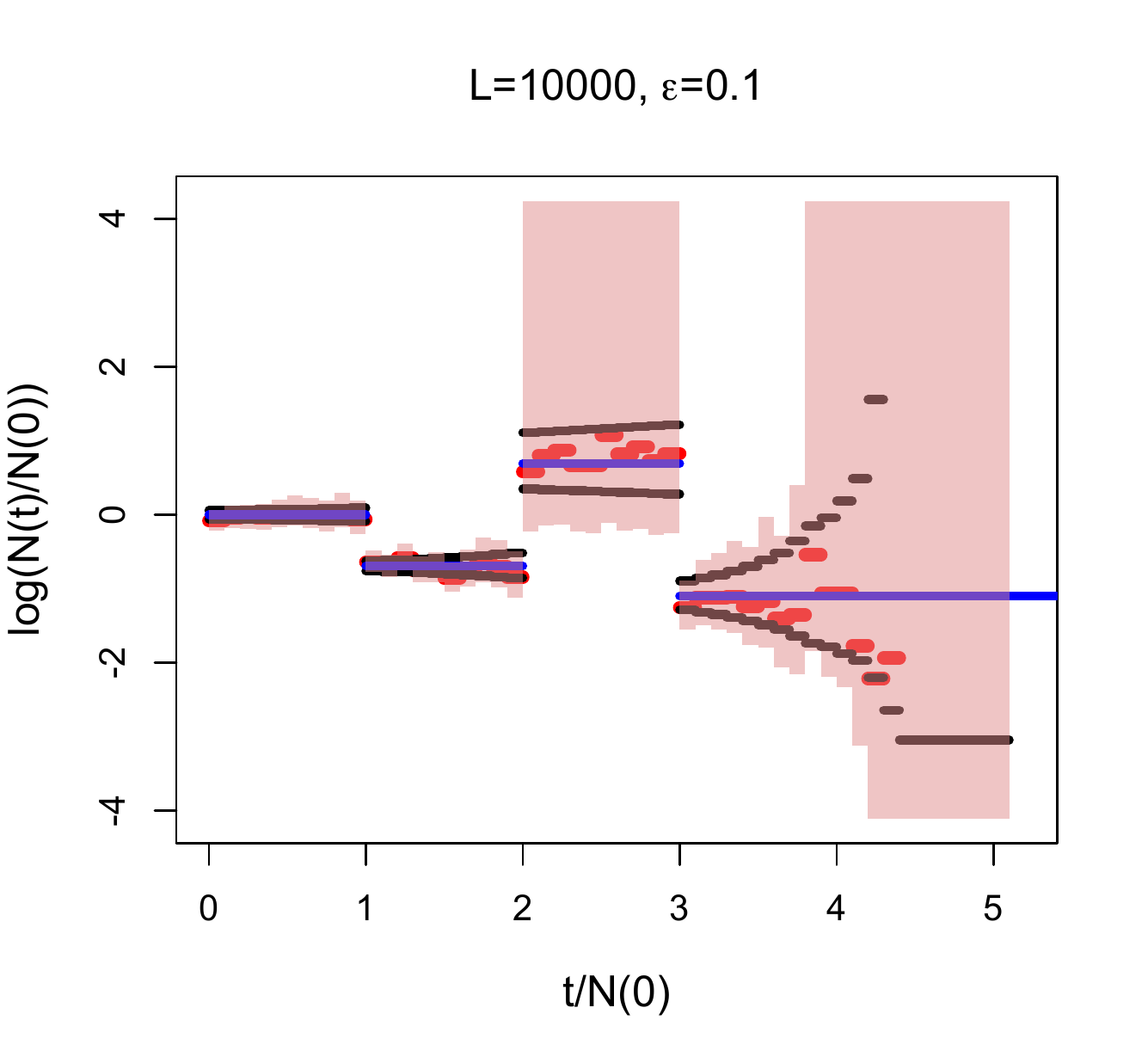}  
	\caption{$L=10^4, \eps = 0.1$.}
	\label{fig:constlb4}
    \end{subfigure}

    \caption{Constant and piecewise constant population histories with uncertainty intervals. }
    \label{fig:reconstruct}
\end{figure}

\end{document}